\documentclass[10pt,a4paper,oneside]{article}

\usepackage{amsmath,amssymb,amsfonts,amsthm}
\usepackage{mathtools}
\usepackage{bm}
\usepackage{algorithm,algorithmic}
\usepackage{tikz}
\usepackage{relsize}
\usepackage{apacite}
\usepackage{graphicx}
\usepackage{subcaption}
\usepackage[ddmmyyyy]{datetime}

\newcommand{\dist}{{\textstyle \mathsmaller{\varDelta}}}

\newtheorem{definition}{Definition}[section]
\newtheorem{proposition}{Proposition}[section]
\newtheorem{lemma}{Lemma}[section]

\title{Randomized Optimal Transport on a Graph:\\
framework and new distance measures\\
\vspace{5pt}
\normalsize{(draft preprint paper accepted for publication in Network Science journal)}}

\author{	Guillaume Guex \\
		{\small ICTEAM, Universit\'{e} catholique de Louvain, Belgium} \and 
		Ilkka Kivim\"aki \\
		{\small Department of Computer Science, Aalto University, Finland \&} \\
		{\small ICTEAM, Universit\'{e} catholique de Louvain, Belgium} \and
		Marco Saerens \\
		{\small ICTEAM, Universit\'{e} catholique de Louvain, Belgium} }

\date{18/02/2019}		

\begin{document}


\maketitle

\begin{abstract}
The recently developed bag-of-paths (BoP) framework consists in setting a Gibbs-Boltzmann distribution on all feasible paths of a graph. This probability distribution favors short paths over long ones, with a free parameter (the temperature $T$) controlling the entropic level of the distribution. This formalism enables the computation of new distances or dissimilarities, interpolating between the shortest-path and the resistance distance, which have been shown to perform well in clustering and classification tasks. In this work, the bag-of-paths formalism is extended by adding two independent equality constraints fixing starting and ending nodes distributions of paths (margins). When the temperature is low, this formalism is shown to be equivalent to a relaxation of the optimal transport problem on a network where paths carry a flow between two discrete distributions on nodes. The randomization is achieved by considering free energy minimization instead of traditional cost minimization. Algorithms computing the optimal free energy solution are developed for two types of paths: hitting (or absorbing) paths and non-hitting, regular, paths, and require the inversion of an $n \times n$ matrix with $n$ being the number of nodes. Interestingly, for regular paths on an undirected graph, the resulting optimal policy interpolates between the deterministic optimal transport policy ($T \rightarrow 0^{+}$) and the solution to the corresponding electrical circuit ($T \rightarrow \infty$). Two distance measures between nodes and a dissimilarity between groups of nodes, both integrating weights on nodes, are derived from this framework. \\

\noindent {\bf Keywords}: Network Science, Optimal Transportation, Bag of Paths, Randomized Shortest Path, Distances between Nodes, Link Analysis.
\end{abstract}


\section{Introduction}

\subsection{General introduction and motivation}
Today, network data are studied in many different areas of science, including applied mathematics, computer science, social science, physics, chemistry, pattern recognition, applied statistics, data mining and machine learning, to name a few (see, e.g., \cite{barabasi2016network,chung2006complex,estrada2012structure,fouss2016algorithms,kolaczyk2009statistical,lewis2009network,newman2010networks,silva2016machine,thelwall2004link,wasserman1994social}).
In this context, one key problem is the definition of distances between nodes taking both direct and indirect connections into account \cite{chebotarev2011class,chebotarev2012walk,chebotarev2013studying,fouss2016algorithms,herbster2009predicting,franccoisse2016bag,lu2011link,alamgir2011phase,yen2008family}.
This problem is faced in many applications such as link prediction, community detection, node classification, and network visualization, among others.

Now, it has been shown that the standard shortest path distance and the resistance distance \cite{klein1993resistance} suffer from important drawbacks in some situations, which sometimes hinders their use as distance measures between nodes. More precisely, the shortest path distance does not integrate the concept of high connectivity between the two nodes (it only considers the shortest paths, see, e.g., \cite{fouss2016algorithms}), while the resistance distance provides useless results when dealing with large graphs (the so-called ``lost-in-space effect" \cite{vonLuxburg2010getting,vonLuxburg2014hitting}).
Another drawback of the shortest path distance is that it provides lots of ties when comparing distances, especially on unweighted undirected graphs.

In order to avoid the drawbacks of the shortest path and resistance distances, new families of distance measures, interpolating between these two extremes, were recently suggested based on a \emph{bag-of-paths} (BoP) \emph{framework} \cite{franccoisse2016bag,kivimaki2014developments,Lebichot-2014,mantrach2009sum}.
This framework defines a Gibbs-Boltzmann probability distribution over paths on a graph, which focuses on the shortest paths, but spreads also on longer paths and random walks.
The spread of the distribution is controlled by a temperature parameter monitoring the balance between choosing low-cost paths and a pure random behaviour.
Different distance measures between nodes are then derived based on this distribution; other ones are described in the next, related work, subsection.

Following this previous work, the effort is pursued in this paper with the introduction of \emph{weighted} distance measures derived from new bag-of-paths (BoP) models.
The weights of the distances are determined by introducing equality constraints on the path distribution margins, i.e., the a priori probabilities over starting nodes and ending nodes of paths.
In other words, the new model assumes that the user knows only where paths on average start and where they on average end but not how they are distributed otherwise.
In the original BoP model, which was developed for the unweighted BoP-based distances, the starting and ending node distributions are instead unconstrained, and can be inferred directly from the probabilities of the paths \cite{franccoisse2016bag}.
The model proposed in the current work will be called the \emph{margin-constrained bag-of-paths framework} (abbreviated as cBoP).
More precisely, the work defines two models by considering two different types of paths -- the first one is based on regular, non-hitting paths, and the second on hitting paths, i.e.\ paths where the ending node cannot appear as an intermediate node.

Weighting the nodes of a network for determining distances can be important in applications where each node represents a whole collection of items (like cities where nodes could be weighted by population). Moreover, in some situations, it could be beneficial to weigh nodes by the reciprocal of their degree in order to avoid the \emph{hubness effect} \cite{radovanovic2010hubs}. This will be investigated experimentally in further work.
In addition to defining weighted distances between nodes on a graph, a dissimilarity measure between groups of nodes on the graph is derived from the margin-constrained BoP framework.

The margin-constrained BoP model can also be understood as defining a \emph{randomized policy} for the \emph{optimal transport problem on a graph} \cite{ahuja1993network,kantorovich1942translocation,villani2003topics,villani2008optimal}, because the starting and ending node distributions can be considered as \emph{supply} and \emph{demand} distributions for goods meant to be transported over the graph.
The randomization is achieved by finding the probability mass on the set of paths connecting starting and ending nodes that minimizes free energy (a balance between expected cost and entropy), subject to margin constraints corresponding to the predefined supply and demand distributions.
As is common for such formulations, minimizing this objective function results in a Gibbs-Boltzmann probability mass on paths.

As discussed in more detail in \cite{saerens2009randomized}, randomization from optimality can prove useful for several reasons, both in the context of transportation, as well as when measuring distance:
\begin{itemize}
\item If the environment is changing over time (non-stationary), the system
could benefit from randomization by performing continual exploration. 
\item A deterministic policy makes behavior totally predictable; on the contrary, randomness introduces unpredictability and therefore renders interception
more difficult. Randomization has proved useful for this reason
in game theory (see, e.g., \cite{Osborne-2004}).
\item A randomized policy spreads the traffic on multiple paths, therefore reducing the danger of congestion. 
\item A distance measure accounting for all paths -- and thus integrating the concept of high connectivity -- can be more useful, e.g.\ in social network analysis, than relying on the best paths only.
\end{itemize}

\subsection{Related work}
The model proposed in this work builds on and extends previous work dedicated to the bag-of-paths (BoP) framework \cite{franccoisse2016bag,mantrach2009sum}, as well as the \emph{randomized-shortest-path} (RSP) framework \cite{akamatsu1996cyclic,kivimaki2014developments,kivimaki2016two,saerens2009randomized,yen2008family}, and their variants \cite{bavaud2012interpolating,guex2015flow,guex2016interpolating}; see also \cite{Zhang-2013} for a related proposition, called path integral.

The main motivation for using such models can be understood as follows \cite{Lebichot-2018}.
Most of the traditional network measures are essentially based on two different paradigms about movement or communication occurring in the network: optimal communication based on shortest paths, and random communication based on a random walk on the graph.
For instance, shortest path distance, as well as the standard betweenness centrality \cite{Freeman-1977} are defined from shortest paths, while resistance distance and random walk centrality \cite{Brandes-2005b,newman2005measure} are based on random walks (which have a strong analogy with electrical flow on the network \cite{Snell-1984}).

However, in reality, communication or movements over a network seldom occur either optimally or purely randomly.
The BoP and RSP frameworks both relax these assumptions by interpolating between shortest paths and a pure random walk based on a temperature parameter. 
This enables the definition of measures with increased adaptability given by the temperature parameter.
In addition to defining distances interpolating between the shortest path and resistance distances, \cite{franccoisse2016bag,kivimaki2014developments,yen2008family}, the models can also be used to define a centrality measure interpolating between a shortest path-based betweenness and the random walk betweenness \cite{kivimaki2016two}.

Besides the works mentioned above, other new families of distances have recently been developed integrating information on both the proximity (shortest-path distance) and amount of connectivity between nodes (captured, e.g., by the resistance distance) \cite{chebotarev2011class,chebotarev2012walk,chebotarev2013studying,fouss2016algorithms,Hashimoto-2015,herbster2009predicting,Li-2013,lu2011link,nguyen2016new,alamgir2011phase}. Many of these measures indeed interpolate (up to a constant scaling factor) between the shortest path distance (or length) and the resistance distance, therefore alleviating the previously mentioned lost-in-space effect.

A short discussion of the standard, deterministic, optimal transport on a graph problem appears in Section \ref{sec::transp}.
Methods based on the optimal transport problem using entropic regularization have recently been investigated in a number of pattern recognition and machine learning tasks (e.g., \cite{Courty-2017,Solomon-2014}). For instance, \cite{Cuturi2013,ferradans2014regularized,guex2017transportation} propose to regularize the standard objective function of the classical discrete optimal transport problem with an entropy term. They show on various problems, including image processing problems, that the resulting algorithm is much faster than the original one. Note that discrete entropy-regularized optimal transport problems were previously studied in economics, transportation science and operations research (see, e.g., \cite{Wilson-1970,Erlander-1990,Kapur-1989,Kapur-1992,Fang-1997}). The main difference with these previous contributions is that the present work defines the different quantities, such as entropy or cost, over \emph{full paths} on the network by adopting a sum-over-paths formalism.

Finally, the hubness effect, mentioned earlier, has been studied recently in various works
\cite{radovanovic2010existence,radovanovic2010hubs,suzuki2012investigating,suzuki2013centering,tomasev2014role,hara2015localized}.
Hubness is a problem faced with high-dimensional data, e.g.\ when building nearest-neighbor graphs, as some nodes may become over-represented as hubs  in such graphs due to concentration of distances in high-dimensional spaces.
The weighting of distances provided by the margin-constrained BoP framework could help alleviate this effect in graph-based data analysis applications.

\subsection{Main contributions}
This work defines weighted distance measures between graph nodes by developing a \emph{margin-constrained} bag-of-paths model.
This model can be interpreted as a solution to the optimal transport problem on a graph involving a regularization term.
The problem is tackled by using Kullback-Leibler divergence (also called relative entropy \cite{Cover-2006}) as regularization term.
Furthermore, two types of paths are considered: regular paths and hitting, absorbing, paths.

The optimal randomized policy consists in the assignment of a probability distribution on the set of choices (deciding to follow an available edge) for each node of the network. It therefore defines (optimal) biased transition probabilities ``attracting" the agents to the destination nodes.
Furthermore, the model depends on a temperature parameter monitoring the balance between \emph{exploitation} (expected cost) and \emph{exploration} (entropy of paths) so that the solutions interpolate between the classical deterministic optimal transport solution (pure exploitation) and the random walk on the graph provided a priori by the user (pure exploration). Low temperatures correspond to (randomized) near-optimal solutions while high temperatures simply provide the predefined random walk behavior.
Note that, when considering hitting paths, the model reduces to the standard randomized shortest path model when there is only one unique initial node and one unique destination node.

The first contribution consists in deriving the probability distribution over paths minimizing expected cost under relative entropy regularization and margin constraints, for both regular and hitting paths.
Once the probability distribution over paths is derived, all the quantities of interest, such as
\begin{itemize}
\item the policy (optimal routing transition probabilities),
\item the flow over the network based on the a priori starting and ending node distributions of paths,
\item a weighted distance measure between nodes, and
\item a dissimilarity between groups of nodes
\end{itemize}
can be defined and computed by simple matrix expressions.

Note that the present work is partly a re-interpretation of \cite{guex2016interpolating} in which the author already studied a similar optimal transport on a graph problem regularized by an entropic term. There, the entropic term at the node level was defined by considering, on each node, the relative entropy between the desired transition probabilities (the policy) and the reference transition probabilities corresponding to a natural random walk on the graph. Then, the global entropic regularization term was defined as a weighted sum of the relative entropies over all nodes. As in \cite{saerens2009randomized}, the weighting factor is set to the expected number of visits to the node, therefore putting more emphasis on frequently visited nodes. In the current work, we adopt a paths-based formalism and the entropic term is instead defined according to the relative entropy over \emph{path distributions}.

Interestingly, it was found that the model derived in \cite{guex2016interpolating} is exactly equivalent to one of the two models introduced in this work, the one dealing with regular, non-hitting, paths, in the sense that they provide the same routing policy. Therefore, in comparison with \cite{guex2016interpolating}, the present work reformulates the problem in terms of probabilities and relative entropy over paths in the network, instead of transition probabilities on nodes. It also introduces another algorithm for solving the problem and it derives a new algorithm for dealing with hitting paths.

%

In short, the main contributions of this paper are
\begin{itemize}
  \item the development of a new \emph{margin-constrained bag-of-paths} framework, considering fixed probability distributions on starting and ending nodes, 
  \item the introduction of a randomized solution to the optimal transport on a graph problem for both regular and hitting paths, 
  \item the definition of a new distance measure between nodes and a dissimilarity between groups of nodes derived from this framework, and
  \item some illustrative simulations to explore the potential of the framework.
\end{itemize}

The remaining of the paper is as follows. Section \ref{Sec_formalism01} develops the formalism and derives the solution to the margin-constrained bag-of-paths problem on a graph for regular paths, while Section \ref{Sec_formalism_hitting_paths01} extends the model to hitting paths. Then, Section \ref{sec::transp} discusses the connections with the standard optimal transport on a graph problem. The derived distances are introduced in Section \ref{Sec_distance01}. Section \ref{Sec_case_studies01} provides some illustrative simulations. Finally, Section \ref{Sec_conclusion01} is the conclusion.


\section{The margin-constrained bag-of-paths formalism}
\label{Sec_formalism01}


\subsection{Background and notation}
\label{Subsec_background01}
This Subsection first sets the notation and terminology of the paper, after which the standard bag-of-paths (BoP) and randomized shortest-paths (RSP) frameworks are briefly reviewed (note that a discussion of the standard optimal transport problem is deferred to Section \ref{sec::transp}).
Then the margin-constrained BoP (cBop) framework, and the relevant related results, are presented.
Note that in this section regular, non-hitting, paths are considered whereas Section \ref{Sec_formalism_hitting_paths01} restricts the set of paths to hitting, or absorbing, paths where the ending node may appear only once as final node.

\paragraph{Notation.}
In this paper, we always assume a weighted, strongly connected, directed graph $\mathcal{G} = (\mathcal{V},\mathcal{E})$, with set of nodes $\mathcal{V} = \{1,2,\ldots,n\}$ and set of edges $\mathcal{E} = \{(i,j)\}$ containing $m$ edges in total. The nonnegative weights on edges, noted $a_{ij}$, represent local affinities between nodes, and are contained in the \emph{weighted adjacency matrix} $\mathbf{A}$. Edge weights define a natural \emph{reference transition probabilities matrix} $\mathbf{P}^\mathrm{ref}$ of a standard random walk on $\mathcal{G}$, with $\mathbf{P}^\mathrm{ref} = \mathbf{D}^{-1}\mathbf{A}$, where $\mathbf{D}$ is the diagonal matrix containing row sums of $\mathbf{A}$ (outdegrees). The Markov chain defined by these transition probabilities is assumed to be regular. Elementwise, we have
\begin{equation}
p^\mathrm{ref}_{ij} = \frac{a_{ij}}{ \sum_{k \in \mathcal{V}} a_{ik} }.
\label{Eq_define_transition_probabilities_elementwise01}
\end{equation}
\noindent Along with weights, nonnegative \emph{edges costs}, noted $c_{ij}$, are also provided. These costs are contained in the \emph{cost matrix} $\mathbf{C}$, and can be defined either independently from weights $a_{ij}$, or, e.g., thanks to $c_{ij} = 1/a_{ij}$. We define a $t$-length \emph{path} on the graph $\mathcal{G}$, denoted by $\wp$, as a sequence of nodes $\wp = (i_0,\dots,i_t)$, where $t>0$ and $(i_{\tau-1},i_\tau) \in \mathcal{E}$ for all $\tau = 1,\dots,t$. Note that a node can appear several times on the path (including the ending node). We denote a path starting in node $i$ and ending in node $j$ by $\wp_{ij}$. The \emph{likelihood} $\tilde{\pi}^\mathrm{ref}(\wp)$ of a $t$-length path $\wp= (i_0,\dots,i_t)$ starting in $i_0$ and ending in $i_t$ is defined by $\tilde{\pi}^\mathrm{ref}(\wp) \triangleq \Pi_{\tau = 1}^t p^\mathrm{ref}_{i_{\tau-1},i_{\tau}}$ and its \emph{cost} $\tilde{c}(\wp)$ by $\tilde{c}(\wp) \triangleq \sum_{\tau = 1}^t c_{i_{\tau-1},i_{\tau}}$. We further denote respectively by $\mathcal{P}_{ij}$ and $\mathcal{P}$, the \emph{set of paths starting in $i$ and ending in $j$} and the  \emph{set of all paths} in $\mathcal{G}$, also named the \emph{bag-of-paths}, with $\mathcal{P} = \cup_{i,j \in \mathcal{V}} \, \mathcal{P}_{ij}$. By convention, zero-length paths starting and ending in the same node with a zero cost are also included in the set of paths (see \cite{franccoisse2016bag} for details). 

All vectors will be column vectors and denoted in lowercase bold while matrices are in uppercase bold.

\paragraph{The bag-of-paths and the randomized-shortest-path frameworks.}

The context defined above states the usual background of the \emph{bag-of-paths} framework developed in \cite{franccoisse2016bag,mantrach2009sum}. In these works, a probability distribution over the set of all paths, $\mathrm{P}(\wp)$  with $\wp \in \mathcal{P}$, was constructed in order to favor paths of low cost subject to a constant relative entropy constraint. It provides the probability of drawing a particular path $\wp$ from a bag, with replacement. The distribution can equivalently be obtained by solving the following problem (see \cite{kivimaki2014developments})\footnote{Alternatively, it can also be obtained by following a maximum entropy argument \cite{Jaynes-1957,Cover-2006,Kapur-1992}.}:
\begin{equation}
\vline \begin{array}{ll@{}ll}
\underset{\{ \mathrm{P}(\wp) \}_{\wp \in \mathcal{P}} }{\text{minimize}}  & \text{FE}(\mathrm{P}) = \displaystyle\sum\limits_{\wp \in \mathcal{P}} \mathrm{P}(\wp) \tilde{c}(\wp) + T \sum_{\wp \in \mathcal{P}} \mathrm{P}(\wp) \log\left( \frac{\mathrm{P}(\wp)}{\mathrm{P^\mathrm{ref}}(\wp)} \right) \\
\text{subject to}& \sum_{\wp \in \mathcal{P}} \mathrm{P}(\wp) = 1,
\end{array} \label{old_prob}
\end{equation}
where $T > 0$, the \emph{temperature}, is a free parameter defining the relative entropy level \cite{Cover-2006}, and $\mathrm{P^\mathrm{ref}}(\wp)$ is the natural, reference, probability of a path $\wp$ depending on $\tilde{\pi}^\mathrm{ref}(\wp)$ (see Equation (\ref{Eq_define_transition_probabilities_elementwise01})) and to be discussed later. $\text{FE}(\mathrm{P})$ is called the (relative) \emph{free energy}, due to its similarity with the statistical physics quantity. It corresponds to the expected cost, or energy, to which the relative entropy weighted by temperature is added. Strictly speaking, a non-negativity constraint should also be added to (\ref{old_prob}), but this is not necessary since the resulting probability distribution will automatically be non-negative.
Indeed, following \cite{Jaynes-1957,Cover-2006,Kapur-1992} and \cite{franccoisse2016bag,kivimaki2014developments,mantrach2009sum} for the paths formalism, the solution is a standard Gibbs-Boltzmann distribution of the form
\begin{equation*}
\mathrm{P}(\wp)
= \frac{\mathrm{P}^{\mathrm{ref}}(\wp)\exp\left[-\theta \tilde{c}(\wp)\right]}{{\displaystyle \sum_{\wp' \in\mathcal{P}}}\mathrm{P}^{\mathrm{ref}}(\wp')\exp\left[-\theta \tilde{c}(\wp')\right]},
\label{Eq_Gibbs_Boltzmann_probability_distribution01}
\end{equation*}
where $\tilde{c}(\wp)$ is the total cumulated cost along path $\wp$. It provides the probability of choosing any path $\wp \in \mathcal{P}$ in the network.

In addition, the \emph{randomized shortest path} framework \cite{saerens2009randomized,kivimaki2014developments,yen2008family}; inspired by \cite{akamatsu1996cyclic} restricts the set of paths to hitting paths (see next section for details) connecting only two predefined nodes $i$ and $j$. This defines optimal randomized policies for reaching node $j$ from $i$, ranging from shortest paths to a random walk. A method for computing the RSP on large sparse graphs by restricting the set to paths with a finite predefined length was developed in \cite[Section 4]{mantrach2011semi}.


\subsection{Problem definition}
\label{Subsec_problemDefinition01}

In this Section, the BoP framework developed in \cite{franccoisse2016bag,mantrach2009sum} is extended into a \emph{margin-constrained bag-of-paths framework} by introducing two additional density vectors on nodes, provided by the user: $\bm{\sigma}_\mathrm{in}$ and $\bm{\sigma}_\mathrm{out}$, with $\sum_{i \in \mathcal{V}} \sigma^\mathrm{in}_i = \sum_{i \in \mathcal{V}} \sigma^\mathrm{out}_i = 1$ and $\sigma^\mathrm{in}_i,\sigma^\mathrm{out}_i \geq 0, \; \forall i \in \mathcal{V}$. These vectors define desired \emph{constraints on the distribution margins} of our bag-of-paths probabilities, i.e.
\begin{eqnarray}
\mathrm{P}(S = i) \triangleq \sum_{j \in \mathcal{V}} \sum_{\wp_{ij} \in \mathcal{P}_{ij}} \mathrm{P}(\wp_{ij}) = \sigma^\mathrm{in}_i & \forall i \in \mathcal{V},  \label{constr_in} \\
\mathrm{P}(E = j) \triangleq \sum_{i \in \mathcal{V}} \sum_{\wp_{ij} \in \mathcal{P}_{ij}} \mathrm{P}(\wp_{ij}) = \sigma^\mathrm{out}_j & \forall j \in \mathcal{V}, \label{constr_out}
\end{eqnarray}
where $S$ and $E$ denote random variables containing respectively the starting and the ending nodes of the drawn path. In turn, it means that we want to constrain the probability of picking a path in the BoP starting from $i$ to value $\sigma^\mathrm{in}_i$ and the probability of picking a path ending in $j$ to value $\sigma^\mathrm{out}_j$. The intuition is as follows: the model assumes that we are carrying a unit of goods in the network from the set of nodes $\mathcal{I}n = \{ i \in \mathcal{V}: \sigma^\mathrm{in}_i > 0 \}$ (supply nodes) to the set $\mathcal{O}ut = \{ j \in \mathcal{V}: \sigma^\mathrm{out}_j > 0 \}$ (demand nodes) in an optimal way by minimizing a balance between expected cost and relative entropy of paths. A discussion of this model in the light of optimal transport on a graph appears later in Section \ref{sec::transp}.

Altogether, we extend problem (\ref{old_prob}) and seek the \emph{optimal paths probability distribution}, $\{ \mathrm{P}(\wp) \}, \wp \in \mathcal{P}$, solving 
\begin{equation}
\label{init_prob}
\vline \begin{array}{ll@{}ll}
\underset{\{ \mathrm{P}(\wp) \}_{\wp \in \mathcal{P}}} {\text{minimize}} & \text{FE}(\mathrm{P}) =  \displaystyle\sum\limits_{\wp \in \mathcal{P}} \mathrm{P}(\wp) \tilde{c}(\wp) + T \sum_{\wp \in \mathcal{P}} \mathrm{P}(\wp) \log \left( \frac{\mathrm{P}(\wp)}{\mathrm{P^\mathrm{ref}}(\wp)} \right) \\
\text{subject to} & \sum_{j \in \mathcal{V}} \sum_{\wp_{ij} \in \mathcal{P}_{ij}} \mathrm{P}(\wp_{ij}) = \sigma^\mathrm{in}_i \qquad \forall i \in \mathcal{V} \\
          & \sum_{i \in \mathcal{V}} \sum_{\wp_{ij} \in \mathcal{P}_{ij}} \mathrm{P}(\wp_{ij}) = \sigma^\mathrm{out}_j \qquad \forall j \in \mathcal{V}. \\ 
\end{array}
\end{equation}
Note that as we have $\sum_{i \in \mathcal{V}} \sigma^\mathrm{in}_i = \sum_{i \in \mathcal{V}} \sigma^\mathrm{out}_i = 1$, the constraint $\sum_{\wp \in \mathcal{P}} \mathrm{P}(\wp) = 1$ in (\ref{old_prob}) can be dropped. The goal of this problem is to find a probability distribution with fixed margins, such that it favors the paths of least cost when $T \to 0$, and the paths with high likelihood $\tilde{\pi}^\mathrm{ref}(\wp)$ when $T \to \infty$. 

In order for the high temperature bounds to be consistent, $\mathrm{P^\mathrm{ref}}(\wp)$ must be defined according to $\tilde{\pi}^\mathrm{ref}(\wp)$. In the usual, unconstrained, BoP formalism with a uniform a priori probability of choosing the starting and ending node, the definition is simply $\mathrm{P^\mathrm{ref}}(\wp) = \tilde{\pi}^\mathrm{ref}(\wp) / \sum_{\wp' \in \mathcal{P}} \tilde{\pi}^\mathrm{ref}(\wp')$ \cite{franccoisse2016bag}\footnote{Note that non-uniform prior probabilities on the starting and ending node are briefly discussed in \cite{franccoisse2016bag}.}. However, defining reference probabilities is not as trivial in the margin-constrained setting studied in this work because of the constraints. This definition is the goal of the next section (note that the reader mainly interested in the randomized optimal transport problem can simply assume that the reference transition probabilities as given, skip the Section \ref{Sec_non_hitting_reference_probabilities01} and go directly to Section \ref{sec:cBoPComputation}). The solution of problem (\ref{init_prob}) is then stated and proved in Section \ref{sec:cBoPComputation}.


\subsection{Reference probabilities with fixed margins}
\label{Sec_non_hitting_reference_probabilities01}

The reference probability of a path, $\mathrm{P}^\mathrm{ref}(\wp)$, should have appropriate margins in order to ensure the convergence $\mathrm{P}^\star(\wp) \to \mathrm{P}^\mathrm{ref}(\wp)$ when $T \to \infty$ (pure random walk). In other words, for consistency, the reference probabilities of paths should also satisfy the constraints,
\begin{eqnarray}
\mathrm{P}^\mathrm{ref}(S = i) = \sum_{j \in \mathcal{V}} \sum_{\wp_{ij} \in \mathcal{P}_{ij}} \mathrm{P}^\mathrm{ref}(\wp_{ij}) = \sigma^\mathrm{in}_i & \forall i \in \mathcal{V},  \label{ref_constr1} \\
\mathrm{P}^\mathrm{ref}(E = j)  = \sum_{i \in \mathcal{V}} \sum_{\wp_{ij} \in \mathcal{P}_{ij}} \mathrm{P}^\mathrm{ref}(\wp_{ij}) = \sigma^\mathrm{out}_j & \forall j \in \mathcal{V}, \label{ref_constr2}
\end{eqnarray}
which further implies that the path probabilities $\mathrm{P}^\mathrm{ref}(\wp_{ij})$ sum to one.

In  \cite{franccoisse2016bag}, the reference distribution of a path is simply set proportional to path likelihood $\tilde{\pi}^\mathrm{ref}(\wp)$, which can be interpreted as follows: the starting distribution $\mathrm{P}^\mathrm{ref}(S = i)$ is defined as uniform and the ending distribution $\mathrm{P}^\mathrm{ref}(E = j)$ is equal to the proportion of time spent in each node for the Markov chain $S_t$, defined by transition probabilities matrix $\mathbf{P}^\mathrm{ref}$, when $t \to \infty$ (stationarity). Obviously, defining the reference probability in a similar way here would lead to a problem: the ending distribution $\mathrm{P}^\mathrm{ref}(E = j)$ depends entirely on the transition matrix of the Markov chain defined by $\mathbf{P}^\mathrm{ref}$ and will generally not yield the desired distribution $\bm{\sigma}_\mathrm{out}$.

To address this problem we introduce a new \emph{killed} Markov process which will lead to the desired ending distribution, while being as similar as possible, in a certain sense, to the original chain. More precisely, we design this killed Markov process in such a way that the random walker encounters exactly the same probabilities of jumping to any adjacent node as the original transition probabilities ($\mathbf{P}^\mathrm{ref}$) as long as he survives.


\subsubsection{A particular killed Markov process}

Let us first define a killed random process which will be helpful later.
\begin{definition}
\label{def::killed_markov_process}
From the reference, regular, Markov process $S_t$ defined on the state space $\mathcal{V}$, with initial distribution $\bm{\pi}_\mathrm{in}$ and transition matrix $\mathbf{P}^\mathrm{ref}$, a \textbf{killed Markov process}, denoted by $K_t$, is defined as a new process with the same initial distribution and following a \textbf{substochastic} transition matrix $\mathbf{\widehat{P}}^\mathrm{ref}$, given by 
\begin{equation}
\mathbf{\widehat{P}}^\mathrm{ref}  \triangleq (\mathbf{I} - \mathbf{Diag}(\bm{\alpha}))\mathbf{P}^\mathrm{ref},
\label{prob_ref_alpha}
\end{equation}
where the \textbf{probability to be killed} after visiting node $i$ is $\alpha_i$, with $0 \leq \alpha_i \leq 1 \; \forall i \in \mathcal{V}$. $\mathbf{Diag}(\bm{\alpha})$ is the diagonal matrix containing vector $\bm{\alpha}$ on its diagonal. In other words, the vector $\bm{\alpha} = (\alpha_i)$ contains the \textbf{killing rate} of each node. This killed Markov process can be seen as adding a new (virtual) absorbing state, the \textbf{cemetery} $\omega$, and following the rules
\begin{align}
\mathrm{P}(K_0 = i) &= \pi^\mathrm{in}_i, \\
\mathrm{P}(K_{t+1} = j | K_{t} = i) &= \widehat{p}_{ij}^\mathrm{ref} = (1 - \alpha_{i}) p_{ij}^{\mathrm{ref}} , \\
\mathrm{P}(K_{t+1} = \omega | K_{t} = i) &= 1 - \widehat{p}_{i \bullet}^\mathrm{ref} \triangleq \alpha_i, \\
\mathrm{P}(K_{t+1} = \omega | K_{t} = \omega) &= 1,
\end{align}
\noindent where the dot in $\widehat{p}_{i \bullet}^\mathrm{ref}$ means summation over the second index (over the set of nodes $\mathcal{V}$).
\end{definition}

We observe that $\widehat{p}^\mathrm{ref}_{ij}/\widehat{p}^\mathrm{ref}_{i\bullet} = p^\mathrm{ref}_{ij}$ for any $\bm{\alpha}$. This means that this killing process will behave similarly to the original process as long as it survives, thus arguing in favor of the similarity requirement between the two chains. However, unlike the reference Markov chain, this killed Markov process possesses an \emph{``ending'' distribution}.

\begin{definition}
\label{def::ending_distribution}
A killed Markov process, as defined in Definition \ref{def::killed_markov_process}, possesses an \textbf{ending distribution}, denoted by $\bm{\pi}_\mathrm{out} = (\pi_i^\mathrm{out})$ and given by
\begin{align}
\pi_i^\mathrm{out} \triangleq \mathrm{P}(K_{M-1} = i) = \sum_{t=0}^\infty \mathrm{P}(K_{t} = i, K_{t+1} = \omega),
\end{align}
where $M$ is the random variable containing the time where the process is killed (it reaches the cemetery state $\omega$).
\end{definition}

This quantity denotes the probability of being killed in state $i$ when starting from initial nodes with probabilities $\pi_i^\mathrm{in}$: it sums up the probability of being killed after $0, 1, 2, \dots$ steps.

Interestingly, it is possible to find the vector of killing rates $\bm{\alpha}$ corresponding to a desired ending distribution $\bm{\pi}_\mathrm{out} = \bm{\sigma}_\mathrm{out}$. This is important as it will allow us to design a proper reference probability distribution satisfying the predefined margins. But we first need the following preliminary lemma.
\begin{lemma}
\label{lemma::number_visits}
The expected number of visits to $j$ before being killed, given that the process started from state $i$, that is, the quantity ${\bar{n}_{ij}^\mathrm{ref}} = \sum_{t=0}^\infty \mathrm{P}(K_{t} = j | K_0 = i)$, is given by element $i$, $j$ of matrix $\mathbf{\widehat{Z}^\mathrm{ref}}  \triangleq \sum_{t=0}^\infty ( \mathbf{\widehat{P}}^\mathrm{ref} )^t = (\mathbf{I} - \mathbf{\widehat{P}}^\mathrm{ref})^{-1}$ which is well-defined for a substochastic matrix and a strongly connected graph. In other words, ${\bar{n}_{ij}^\mathrm{ref}} = {\widehat{z}_{ij}^\mathrm{ref}}$. 
\end{lemma}
\begin{proof}
See \cite[Section 1.5.7]{fouss2016algorithms}.
\end{proof}

Then, the killing rates and the ending distribution are related by the following proposition. This will allow us to determine the killing rates in order to satisfy a predefined ending distribution.
\begin{proposition}
\label{prop::piout_to_alpha}
For a killed Markov process on a strongly connected graph (see Definition \ref{def::killed_markov_process}), the following equality is satisfied
\begin{equation}
\bm{\alpha} = \bm{\pi}_\mathrm{out} \div \mathbf{\bar{n}^\mathrm{ref}},
\label{a_from_n}
\end{equation}
where $\div$ is the elementwise division and column vector $\mathbf{\bar{n}^\mathrm{ref}} = ({\bar{n}_j^\mathrm{ref}})$ with ${\bar{n}_j^\mathrm{ref}} = \sum_{i \in \mathcal{V}} \pi^\mathrm{in}_i \bar{n}_{ij}^\mathrm{ref}$, holds the expected number of visits to state $i$ before being killed which can be computed thanks to 
\begin{equation}
\big( \mathbf{I} - (\mathbf{P}^{\mathrm{ref}})^\top \big) \mathbf{\bar{n}^\mathrm{ref}} = \bm{\pi}_\mathrm{in} - (\mathbf{P}^{\mathrm{ref}})^\top \bm{\pi}_\mathrm{out}, \label{Eq_number_of_visits}
\end{equation}
where $\mathbf{P}^{\mathrm{ref}}$ is the transition probability matrix defined in Equation (\ref{Eq_define_transition_probabilities_elementwise01}).
\end{proposition}

\begin{proof}
First, let us observe that the joint distribution for the starting node and the ending node of the killing process is  
\begin{align}
\mathrm{P}(K_0 = i, K_{M-1} = j) &= \sum_{t=0}^\infty \mathrm{P}(K_0 = i, K_{t} = j, K_{t+1} = \omega) \nonumber \\
&= \sum_{t=0}^\infty \mathrm{P}(K_{t+1} = \omega | K_{t} = j) \mathrm{P}(K_{t} = j | K_0 = i) \mathrm{P}(K_0 = i) \nonumber \\
&= \pi^\mathrm{in}_{i}  \left[ \sum_{t=0}^\infty ( \mathbf{\widehat{P}}^\mathrm{ref} )^t \right]_{ij} \alpha_{j}
= \pi^\mathrm{in}_{i}  \left[ \mathbf{\widehat{Z}^\mathrm{ref}} \right]_{ij} \alpha_{j} \nonumber \\
&= \pi^\mathrm{in}_{i} \widehat{z}_{ij}^\mathrm{ref} \alpha_{j},
\label{Eq_joint_probability_killed01}
\end{align}
where $\mathbf{\widehat{Z}^\mathrm{ref}}$ is defined by Lemma \ref{lemma::number_visits}. Thus the ending distribution (see Definition \ref{def::ending_distribution}) reads
\begin{gather}
\pi^\mathrm{out}_j = \sum_{i \in \mathcal{V}} \mathrm{P}(K_0 = i, K_{M-1} = j) =  \sum_{i \in \mathcal{V}} \pi^\mathrm{in}_{i} \widehat{z}^\mathrm{ref}_{ij} \alpha_{j}. \label{A_fe}
\end{gather}
We know from Lemma \ref{lemma::number_visits} that $\widehat{z}^\mathrm{ref}_{ij}$ can be interpreted as the expected number of times the killed Markov process visits node $j$ when starting from $i$. So if we define  
\begin{equation}
{\bar{n}_j^\mathrm{ref}} \triangleq \sum_{i \in \mathcal{V}} \pi^\mathrm{in}_i \widehat{z}^\mathrm{ref}_{ij} = \sum_{i \in \mathcal{V}} \pi^\mathrm{in}_i \bar{n}_{ij}^\mathrm{ref}  \label{n_ref_def},
\end{equation}
the column vector $\mathbf{\bar{n}^\mathrm{ref}} = ({\bar{n}_j^\mathrm{ref}})$ holds the expected number of times the process is in $j$ before being killed. Then, (\ref{A_fe}) directly provides
\begin{equation}
\bm{\alpha} = \bm{\pi}_\mathrm{out} \div \mathbf{\bar{n}^\mathrm{ref}}.
\label{a_from_n_bis}
\end{equation}
The second part of the proposition is obtained by starting from the definition of $\mathbf{\bar{n}^\mathrm{ref}}$ and following (\ref{prob_ref_alpha}),
\begin{align*}
\mathbf{\bar{n}^\mathrm{ref}}& = (\mathbf{\widehat{Z}^\mathrm{ref}})^\top \bm{\pi}_\mathrm{in} = \big( \mathbf{I} - (\mathbf{\widehat{P}}^{\mathrm{ref}})^\top \big)^{-1} \bm{\pi}_\mathrm{in},   \\
\Rightarrow \quad \mathbf{\bar{n}^\mathrm{ref}}&  - (\mathbf{\widehat{P}}^{\mathrm{ref}})^\top \mathbf{\bar{n}^\mathrm{ref}} = \bm{\pi}_\mathrm{in}, \\
\Rightarrow \quad \mathbf{\bar{n}^\mathrm{ref}}&  - \big( (\mathbf{P}^{\mathrm{ref}})^\top - (\mathbf{P}^{\mathrm{ref}})^\top \mathbf{Diag}(\bm{\alpha}) \big) \mathbf{\bar{n}^\mathrm{ref}} = \bm{\pi}_\mathrm{in}.
\end{align*}
Then, using $\mathbf{Diag}(\bm{\alpha}) \, \mathbf{\bar{n}^\mathrm{ref}} = \bm{\pi}_\mathrm{out}$ (see Equation (\ref{a_from_n})),
\begin{equation*}
\big( \mathbf{I} - (\mathbf{P}^{\mathrm{ref}})^\top \big) \mathbf{\bar{n}^\mathrm{ref}} = \bm{\pi}_\mathrm{in} - (\mathbf{P}^{\mathrm{ref}})^\top \bm{\pi}_\mathrm{out}, 
\end{equation*}
which provides the same results as the expression derived in \cite[proposition 1]{guex2016interpolating} from another perspective.
\end{proof}
\noindent The Equation (\ref{a_from_n}) simply states that the probability of being killed in state $i$ is equal to the expected number of visits to $i$ times the probability of jumping to the cemetery state from $i$, $\pi^{\mathrm{out}}_{i} = \bar{n}_{i}^{\mathrm{ref}} \, \alpha_{i}$. This is similar to the computation of the absorption probabilities when starting from a transient state in an absorbing Markov chain \cite{Snell-1984,Grinstead-1997}.

In conclusion, it is possible to determine $\mathbf{\bar{n}^\mathrm{ref}}$ and then $\bm{\alpha}$ through Proposition \ref{prop::piout_to_alpha} by considering an additional free parameter.  Indeed, as $(\mathbf{I} - (\mathbf{P}^{\mathrm{ref}})^\top)$ is rank-deficient (its rank is $n-1$, as the initial reference chain is regular), we have (see, e.g., \cite{Graybill-1983})
\begin{equation}
\mathbf{\bar{n}^\mathrm{ref}} = \big( \mathbf{I} - (\mathbf{P}^{\mathrm{ref}})^\top \big)^{+} \big( \bm{\pi}_\mathrm{in} - (\mathbf{P}^{\mathrm{ref}})^\top \bm{\pi}_\mathrm{out} \big) + \epsilon \, \bm{\pi}_\text{ref},
\label{Eq_computing_number_of_visits01}
\end{equation}
where $\big( \mathbf{I} - (\mathbf{P}^{\mathrm{ref}})^\top \big)^{+}$ denotes the Moore-Penrose pseudoinverse, $\bm{\pi}_\text{ref}$ is the stationary distribution of the regular, reference, Markov chain defined by $\mathbf{P}^\mathrm{ref}$ (i.e. a vector summing to 1 and generating the null-space of $( \mathbf{I} - (\mathbf{P}^{\mathrm{ref}})^\top)$), and $\epsilon$ is an additional free parameter, named \emph{persistence}, such that
\begin{equation*}
\epsilon \geq \max_i \left( \frac{ \pi^\mathrm{out}_i - \mathbf{e}_{i}^\top \big( \mathbf{I} - (\mathbf{P}^{\mathrm{ref}})^\top \big)^{+} \big(\bm{\pi}_\mathrm{in} - (\mathbf{P}^{\mathrm{ref}})^\top \bm{\pi}_\mathrm{out} \big)  }{\pi^\mathrm{ref}_i} \right).
\end{equation*}
This last inequality ensures that $\bar{n}^\mathrm{ref}_i \geq \pi^\mathrm{out}_i$ and thus $\alpha_i \leq 1, \, \forall i \in \mathcal{V}$. Intuitively, the persistence parameter reflects the difficulty for a process to be killed in nodes where $\alpha_i>0$ (called killing nodes), and thus affects the expected length of the paths (see \cite{guex2016interpolating} for a discussion). In \cite{guex2016interpolating}, the persistence is shown to have an electrical interpretation within the well-known analogy between random-walk models and electrical models \cite{Snell-1984}: it corresponds to the lowest electrical potential that can be defined on nodes. However, the effect of the persistence on the behavior of the model is beyond the scope of this work, and we will set it to its lower bound in our simulations.

\begin{algorithm}[t!]
\caption[Killing probabilites]
{Killing rates for non-hitting paths}
\algsetup{indent=2em, linenodelimiter=.}
\begin{algorithmic}[1]
\small
\REQUIRE $\,$ \\
 -- The $n\times n$ non-negative adjacency matrix $\mathbf{A}$ of a strongly connected directed graph, containing edge affinities.\\
 -- An input distribution vector $\bm{\sigma}_\mathrm{in}$ of size $n$.\\
 -- An output distribution vector $\bm{\sigma}_\mathrm{out}$ of size $n$.\\
 -- A scalar persistence gap parameter $\epsilon_{\mathrm{g}} > 0$. \\
 
\ENSURE $\,$ \\
 -- The killing rates vector $\bm{\alpha}$ of size $n$. \\
 -- The expected number of visits to each node for the reference random walk, $\mathbf{\bar{n}^\mathrm{ref}}$. \\

~\\

\STATE $\mathbf{D} \leftarrow \mathbf{Diag}(\mathbf{A}\mathbf{e})$ \COMMENT{the out-degree matrix} \\
\STATE $\mathbf{P}^\mathrm{ref} \leftarrow \mathbf{D}^{-1} \mathbf{A}$ \COMMENT{the reference transition probabilities matrix} \\
\STATE $\mathbf{Q} \leftarrow \mathbf{I} - (\mathbf{P}^\mathrm{ref})^\top$ \COMMENT{the auxiliary matrix $\mathbf{Q}$} \\
\STATE $\bm{\pi} \leftarrow \mathbf{null}(\mathbf{Q}) $; $\bm{\pi} \leftarrow \bm{\pi}/\mathrm{sum}(\bm{\pi})$ \COMMENT{computes the stationary distribution of the Markov chain whose transition matrix is $\mathbf{P}^\mathrm{ref}$} ($\bm{\pi}$ spans the null space of $\mathbf{Q}$ and its elements are non-negative and sum to one) \\
\STATE $\mathbf{Q}^{+} \leftarrow \text{pinv}(\mathbf{Q})$ \COMMENT{the Moore-Penrose pseudoinverse matrix of $\mathbf{Q}$} \\
\STATE $\mathbf{\bar{n}_0^\mathrm{ref}} \leftarrow \mathbf{Q}^{+}(\bm{\sigma}_\mathrm{in} - (\mathbf{P}^\mathrm{ref})^\top \bm{\sigma}_\mathrm{out})$ \COMMENT{the initial expected number of visits to each node for the reference random walk} \\
\STATE $\epsilon \leftarrow \max( (\bm{\sigma}_\mathrm{out} - \mathbf{\bar{n}_0^\mathrm{ref}}) \div \bm{\pi}) + \epsilon_{\mathrm{g}}$ \COMMENT{the value of the persistence parameter} \\
\STATE $\mathbf{\bar{n}^\mathrm{ref}} \leftarrow \mathbf{\bar{n}_0^\mathrm{ref}} + \epsilon \bm{\pi}$ \COMMENT{the expected number of visits to each node for the reference random walk} \\
\STATE $\bm{\alpha} \leftarrow \bm{\sigma}_\mathrm{out} \div \mathbf{\bar{n}^\mathrm{ref}}$ \COMMENT{the killing rates vector; $\div$ is the elementwise division} \\

\RETURN $\bm{\alpha}$, $\mathbf{\bar{n}^\mathrm{ref}} $
\end{algorithmic}
\label{Alg_kill_prob} 
\end{algorithm}



\subsubsection{Defining the reference probabilities over paths from the killed process}

We will now define the reference probabilities $\mathrm{P}^\mathrm{ref}(\wp)$ over the set of paths, with the help of a killed Markov process, in order to obtain the desired starting and ending distributions $\bm{\sigma}_\mathrm{in}$ and $\bm{\sigma}_\mathrm{out}$, that is, $\sum_{j \in \mathcal{V}} \sum_{\wp_{ij} \in \mathcal{P}_{ij}} \mathrm{P}^\mathrm{ref}(\wp_{ij}) = \sigma^\mathrm{in}_i$ and $\sum_{i \in \mathcal{V}} \sum_{\wp_{ij} \in \mathcal{P}_{ij}} \mathrm{P}^\mathrm{ref}(\wp_{ij}) = \sigma^\mathrm{out}_j$. We show that it suffices to set $\bm{\pi}_\mathrm{in} = \bm{\sigma}_\mathrm{in}$ and $\bm{\pi}_\mathrm{out} = \bm{\sigma}_\mathrm{out}$ in the previously defined killed Markov process and deduce the reference probabilities from it.
\begin{proposition}
\label{prop::reference_probabiliies}
If the reference probabilities $\mathrm{P}^\mathrm{ref}(\wp)$, $\wp \in \mathcal{P}$, are set to
\begin{equation}
\mathrm{P}^\mathrm{ref}(\wp_{ij}) \triangleq \sigma^\mathrm{in}_i \widehat{\pi}^\mathrm{ref}(\wp_{ij}) \, \alpha_j \text{ for each } i,j \in \mathcal{V}, \label{prob_ref}
\end{equation}
where $\widehat{\pi}^\mathrm{ref}(\wp_{ij})$ is the modified likelihood of the path $\wp_{ij}$, defined by
\begin{equation}
\widehat{\pi}^\mathrm{ref}(\wp)  \triangleq \prod_{\tau = 1}^{t} \widehat{p}^\mathrm{ref}_{k_{\tau-1} k_\tau} \text{ for any path } \wp = (k_{0}, k_{1}, \dots, k_{t}), 
\end{equation}
with $\mathbf{\widehat{P}}^\mathrm{ref} = (\mathbf{I} - \mathbf{Diag}(\bm{\alpha}))\mathbf{P}^\mathrm{ref}$, and where $\bm{\alpha}$ is computed from Proposition \ref{prop::piout_to_alpha} (with $\bm{\pi}_\mathrm{in} = \bm{\sigma}_\mathrm{in}$ and $\bm{\pi}_\mathrm{out} = \bm{\sigma}_\mathrm{out}$), then we have 
\begin{eqnarray*}
\mathrm{P}^\mathrm{ref}(S = i) = \sum_{j \in \mathcal{V}} \sum_{\wp_{ij} \in \mathcal{P}_{ij}} \mathrm{P}^\mathrm{ref}(\wp_{ij}) = \sigma^\mathrm{in}_i & \forall i \in \mathcal{V},  \\
\mathrm{P}^\mathrm{ref}(E = j)  = \sum_{i \in \mathcal{V}} \sum_{\wp_{ij} \in \mathcal{P}_{ij}} \mathrm{P}^\mathrm{ref}(\wp_{ij}) = \sigma^\mathrm{out}_j & \forall j \in \mathcal{V}.
\end{eqnarray*}
Moreover, these reference probabilities are properly scaled as they sum to one.
\end{proposition}

\begin{proof}
Observe that if we set the reference probabilities to
$\mathrm{P}^\mathrm{ref}(\wp_{ij}) \triangleq \pi^\mathrm{in}_i \widehat{\pi}^\mathrm{ref}(\wp_{ij}) \, \alpha_j$,
then, from (\ref{Eq_joint_probability_killed01}), we have for the probability of picking a paths $\wp_{ij} \in \mathcal{P}_{ij}$ in the bag of all paths starting in $i$ ($S=i$) and ending in $j$ ($E = j$) (see \cite{franccoisse2016bag})
\begin{align*}
\mathrm{P}(S = i, E = j) &= \sum_{\wp_{ij} \in \mathcal{P}_{ij}} \mathrm{P}^\mathrm{ref}(\wp_{ij}) = \pi^\mathrm{in}_{i}  \sum_{\wp_{ij} \in \mathcal{P}_{ij}} \widehat{\pi}^\mathrm{ref}(\wp_{ij}) \, \alpha_{j} \\
&=  \pi^\mathrm{in}_{i} \left[ \sum_{t=0}^\infty (\mathbf{\widehat{P}}^\mathrm{ref})^t \right]_{ij} \alpha_{j} = \mathrm{P}(K_0 = i, K_{M-1} = j),
\end{align*}
and therefore,
\begin{align*}
\mathrm{P}(S = i) &=  \mathrm{P}(K_0 = i) = \pi^\mathrm{in}_i , \\
\mathrm{P}(E = j) &=  \mathrm{P}( K_{M-1} = j) = \pi^\mathrm{out}_j .
\end{align*}

Moreover, it also follows that $\sum_{\wp \in \mathcal{P}} \mathrm{P}^\mathrm{ref}(\wp) = \sum_{i,j=1}^{n} \sum_{\wp_{ij} \in \mathcal{P}_{ij}} \mathrm{P}^\mathrm{ref}(\wp_{ij}) = \sum_{i,j=1}^{n} \mathrm{P}(K_0 = i, K_{M-1} = j) = 1$ so that the reference probability distribution over paths is properly scaled.
\end{proof}
Note that the last quantity, $\mathrm{P}(S = i, E = j)$, was called the bag-of-paths probability matrix and played a key role in the bag-of-paths framework (see \cite{franccoisse2016bag} for details). It is also called the coupling matrix in optimal transportation (see later).
The procedure for computing the vectors $\bm{\alpha}$ and $\mathbf{\bar{n}^\mathrm{ref}}$ allowing to obtain a desired ending distribution is sumarized in Algorithm \ref{Alg_kill_prob}. In addition to the predefined margins, it takes as input a persistence gap parameter $\epsilon_{\mathrm{g}}$ indicating to which extend persistence of flow is present in the network (see the discussion following Equation (\ref{Eq_computing_number_of_visits01})).


\subsection{Computation of the optimal probability distribution over paths}
\label{sec:cBoPComputation}
In this section, now that we have found a proper reference distribution, we focus on the computation of the optimal probability distribution $\mathrm{P}^\star(\wp)$ solving problem (\ref{init_prob}). This solution is obtained through its Lagrange parameter vectors, which can be obtained from the constraints.

\subsubsection{The optimal path probabilities}

The optimal probability distribution is obtained by the following proposition.

\begin{proposition}
\label{prop::optimal_sol}
For reference path probabilities $\mathrm{P^\mathrm{ref}}(\wp)$ as defined in Proposition \ref{prop::reference_probabiliies}, the minimization problem $(\ref{init_prob})$, as stated in Section \ref{Subsec_problemDefinition01}, is solved by
\begin{equation}
\mathrm{P}^\star(\wp_{ij}) = \mu^\mathrm{in}_{i}  \sigma^\mathrm{in}_{i} \mu^\mathrm{out}_{j}   \alpha_{j}  \widehat{\pi}^\mathrm{ref}(\wp_{ij}) \exp(-\beta \tilde{c}(\wp_{ij} ))  \text{ for each } i,j \in \mathcal{V}, 
\label{prob_ref_comp}
\end{equation}
where $\beta \triangleq 1 / T$ is the inverse temperature parameter, and $\mu^\mathrm{in}_i \triangleq \exp( - \beta \lambda^\mathrm{in}_i)$, $\mu^\mathrm{out}_j \triangleq \exp( - \beta \lambda^\mathrm{out}_j)$ are two vectors derived from the Lagrange parameter vectors $\bm{\lambda}_\mathrm{in}$, $\bm{\lambda}_\mathrm{out}$, associated with the constraints.
\end{proposition}
\begin{proof}
We derive the solution for the optimal probability distribution $\mathrm{P}^\star(\wp)$ solving problem (\ref{init_prob}). By introducing Lagrange parameter vectors $\bm{\lambda}_\mathrm{in}$ and $\bm{\lambda}_\mathrm{out}$, the Lagrange function associated to (\ref{init_prob}) is
\begin{align}
\mathcal{L}(\mathrm{P},\bm{\lambda}_\mathrm{in},\bm{\lambda}_\mathrm{out})
&\triangleq \sum_{\wp \in \mathcal{P}} \mathrm{P}(\wp) \tilde{c}(\wp) + T  \sum_{\wp \in \mathcal{P}} \mathrm{P}(\wp) \log \left( \frac{\mathrm{P}(\wp)}{\mathrm{P^\mathrm{ref}}(\wp)} \right) \notag \\ 
&+ \sum_{i \in \mathcal{V}} ( \lambda^\mathrm{in}_i - \tfrac{T}{2} ) \left[ \sum_{j \in \mathcal{V}} \sum_{\wp_{ij} \in \mathcal{P}_{ij}} \mathrm{P}(\wp_{ij}) - \sigma^\mathrm{in}_i \right] \notag \\ 
&+ \sum_{j \in \mathcal{V}} ( \lambda^\mathrm{out}_j - \tfrac{T}{2} ) \left[ \sum_{i \in \mathcal{V}} \sum_{\wp_{ij} \in \mathcal{P}_{ij}} \mathrm{P}(\wp_{ij}) - \sigma^\mathrm{out}_j \right] \label{lagrangian},
\end{align}
where the Lagrange parameters are shifted by $T/2$ to simplify the notation.
Taking its partial derivative with respect to $\mathrm{P}(\wp_{ij})$, setting the result to zero, and defining the \emph{inverse temperature} $\beta \triangleq 1/T$, provides
\begin{gather*}
\mathrm{P}^\star(\wp_{ij}) = \mathrm{P}^\mathrm{ref}(\wp_{ij}) \exp\left[ - \beta \left( \tilde{c}(\wp_{ij}) + \lambda^\mathrm{in}_{i} + \lambda^\mathrm{out}_{j} \right) \right] \\
= \mu^\mathrm{in}_{i} \mu^\mathrm{out}_{j} \mathrm{P}^\mathrm{ref}(\wp_{ij}) \exp(-\beta \tilde{c}(\wp_{ij})), 
\end{gather*}
where we defined 
\begin{equation}
\mu^\mathrm{in}_i \triangleq \exp( - \beta \lambda^\mathrm{in}_i) \text{ and } \mu^\mathrm{out}_i \triangleq \exp( - \beta \lambda^\mathrm{out}_i),
\label{Eq_lagrange_parameters_mu01}
\end{equation}
which corresponds to a re-parametrization of the the Lagrange parameters that will be used instead of the original parameters.
By inserting the reference probability found in (\ref{prob_ref}) in this last equation, we get the following form for the path probabilities
\begin{equation*}
\mathrm{P}^\star(\wp_{ij}) = \mu^\mathrm{in}_{i}  \sigma^\mathrm{in}_{i} \mu^\mathrm{out}_{j}   \alpha_{j}  \widehat{\pi}^\mathrm{ref}(\wp_{ij}) \exp(-\beta \tilde{c}(\wp_{ij} )),
\end{equation*}
where the $\widehat{\pi}^\mathrm{ref}(\wp)$ are provided by Proposition \ref{prop::reference_probabiliies}.
\end{proof}

\subsubsection{Computing the Lagrange parameters}

The solution (\ref{prob_ref_comp}) requires the values of the Lagrange multipliers $\bm{\lambda}_\mathrm{in}$ and $\bm{\lambda}_\mathrm{out}$, or alternatively the vectors $\bm{\mu}_\mathrm{in}$ and $\bm{\mu}_\mathrm{out}$, which can be obtained from the equality constraints. Proposition \ref{prop::dual_vectors} shows  how to compute these vectors.
\begin{proposition} 
\label{prop::dual_vectors}
The two vectors, defined as $\mu^\mathrm{in}_i = \exp( - \beta \lambda^\mathrm{in}_i)$ and $\mu^\mathrm{out}_i = \exp( - \beta \lambda^\mathrm{out}_i)$ (see Proposition \ref{Subsec_problemDefinition01}), verify the following equations
\begin{eqnarray}
\bm{\mu}_\mathrm{in} = \mathbf{e} \div  \left( \mathbf{\widehat{Z}}\left( \bm{\mu}_\mathrm{out} \circ \bm{\alpha} \right) \right), \label{it1} \\
\bm{\mu}_\mathrm{out} = \bar{\mathbf{n}}^\mathrm{ref} \div  \left( \mathbf{\widehat{Z}}^\top \left( \bm{\mu}_\mathrm{in} \circ \bm{\sigma}_\mathrm{in} \right) \right), \label{it2}
\end{eqnarray} 
where $\mathbf{\widehat{Z}} \triangleq (\mathbf{I} - \widehat{\mathbf{W}})^{-1}$ is the \textbf{fundamental matrix}, obtained from $\widehat{\mathbf{W}} \triangleq \mathbf{\widehat{P}}^\mathrm{ref} \circ \exp[-\beta \mathbf{C}]$, and $\circ$ is the elementwise product.
\end{proposition}
\begin{proof}
The Lagrange parameters can be found by enforcing the constraints (\ref{constr_in}) and (\ref{constr_out}) on $\mathrm{P}^\star(\wp)$. By injecting (\ref{prob_ref_comp}) for $\mathrm{P}^\star(\wp)$ in (\ref{constr_in}) and (\ref{constr_out}) provides
\begin{eqnarray}
\sigma^\mathrm{in}_i \mu^\mathrm{in}_i \sum_{j \in \mathcal{V}} \mu^\mathrm{out}_j \alpha_j \underbracket[0.5pt][4pt]{ \sum_{\wp_{ij} \in \mathcal{P}_{ij}} \widehat{\pi}^\mathrm{ref}(\wp_{ij}) \exp( - \beta \tilde{c}(\wp_{ij})) }_{\widehat{z}_{ij}} = \sigma^\mathrm{in}_i  & \forall i \in \mathcal{V}, \label{it1_elem} \\
\alpha_j \mu^\mathrm{out}_j \sum_{i \in \mathcal{V}} \mu^\mathrm{in}_i \sigma^\mathrm{in}_i \underbracket[0.5pt][4pt]{ \sum_{\wp_{ij} \in \mathcal{P}_{ij}} \widehat{\pi}^\mathrm{ref}(\wp_{ij}) \exp( - \beta \tilde{c}(\wp_{ij})) }_{\widehat{z}_{ij}} = \sigma^\mathrm{out}_j & \forall j \in \mathcal{V}.
\label{it2_elem}
\end{eqnarray}
By further defining the \emph{fundamental matrix} as  
\begin{equation}
\mathbf{\widehat{Z}} \triangleq (\mathbf{I} - \widehat{\mathbf{W}})^{-1}
\label{fundamental_matrix}
\end{equation}
with $\widehat{\mathbf{W}} \triangleq \mathbf{\widehat{P}}^\mathrm{ref} \circ \exp[-\beta \mathbf{C}]$, where $\circ$ is the elementwise product, it can easily be shown by using a development similar to \cite{franccoisse2016bag,mantrach2009sum} that
\begin{equation}
\sum_{\wp_{ij} \in \mathcal{P}_{ij}} \widehat{\pi}^\mathrm{ref}(\wp_{ij}) \exp( - \beta \tilde{c}(\wp_{ij})) = \widehat{z}_{ij} = [\mathbf{\widehat{Z}}]_{ij}.
\label{Eq_fundamental_matrix_regular_paths01}
\end{equation}

For computing the parameters, we use a variant \cite{Kapur-1992} of iterative proportional fitting procedure (discussed below) based on (\ref{it1_elem}) and (\ref{it2_elem}). Isolating $\sigma^\mathrm{in}_i$ and $\sigma^\mathrm{out}_i$ in (\ref{it1_elem}) and (\ref{it2_elem}) after replacing $\alpha_j$ by the result found with Proposition \ref{prop::piout_to_alpha} (Equations (\ref{a_from_n}) and (\ref{Eq_computing_number_of_visits01})), i.e. $\alpha_j  = \sigma^\mathrm{out}_j / \bar{n}^\mathrm{ref}_j$, in the second equation, we obtain
\begin{eqnarray*}
\bm{\mu}_\mathrm{in} = \mathbf{e} \div  \left( \mathbf{\widehat{Z}}\left( \bm{\mu}_\mathrm{out} \circ \bm{\alpha} \right) \right), \\
\bm{\mu}_\mathrm{out} = \bar{\mathbf{n}}^\mathrm{ref} \div  \left( \mathbf{\widehat{Z}}^\top \left( \bm{\mu}_\mathrm{in} \circ \bm{\sigma}_\mathrm{in} \right) \right).
\end{eqnarray*} 
Note that the above derivation is only valid for nodes $i$ and $j$ for which $\sigma^\mathrm{in}_i > 0$ and $\sigma^\mathrm{out}_j > 0$, respectively, and that $\mu^\mathrm{in}_i$ and $\mu^\mathrm{out}_j$ are not needed when these quantities are equal to zero because the path probabilities also vanish in this case (see Equations (\ref{prob_ref_comp})). However, defining and computing the quantities $\mu^\mathrm{in}_i$ and $\mu^\mathrm{out}_j$ for all nodes $i$ and $j$ according to (\ref{it1}) and (\ref{it2}) proves to be convenient for what follows, as these quantities appear in other meaningful expressions. For unconstrained nodes with $\sigma^\mathrm{in}_i = 0$ or $\sigma^\mathrm{out}_j = 0$, the Lagrange parameters are equal to zero, meaning that the corresponding $\mu^\mathrm{in}_i = 1$ and $\mu^\mathrm{out}_j = 1$
\end{proof}

Note that by further defining the matrix $\widehat{t}_{ij} \triangleq \sigma^\mathrm{in}_i \widehat{z}_{ij} \alpha_j$, (\ref{it1_elem}) and (\ref{it2_elem}) can be rewritten in matrix form as $\mathbf{Diag}(\bm{\mu}_\mathrm{in}) ( \mathbf{Diag}(\bm{\sigma}_\mathrm{in}) \mathbf{\widehat{Z}} \, \mathbf{Diag}(\bm{\alpha)} ) \mathbf{Diag}(\bm{\mu}_\mathrm{out}) \, \mathbf{e}$ $=$ $\mathbf{Diag}(\bm{\mu}_\mathrm{in}) \widehat{\mathbf{T}} \, \mathbf{Diag}(\bm{\mu}_\mathrm{out}) \, \mathbf{e}$ $=$ $\bm{\sigma}_\mathrm{in}$ and $\mathbf{Diag}(\bm{\mu}_\mathrm{out}) \widehat{\mathbf{T}}^\top \mathbf{Diag}(\bm{\mu}_\mathrm{in}) \, \mathbf{e}$ $=$ $\bm{\sigma}_\mathrm{out}$, where $\mathbf{Diag}(\mathbf{x})$ is a diagonal matrix with $\mathbf{x}$ on its main diagonal and $\mathbf{e}$ is a column vector full of $1$'s. Therefore, we have in matrix form
\begin{equation}
\widehat{\mathbf{T}} = \mathbf{Diag}(\bm{\sigma}_\mathrm{in}) \mathbf{\widehat{Z}} \, \mathbf{Diag}(\bm{\alpha)}.
\label{Eq_T_matrix01}
\end{equation}
Thus, the computation of the Lagrange parameters reduces to the problem of finding two nonnegative row and column scaling vectors ($\bm{\mu}_\mathrm{in}$ and $\bm{\mu}_\mathrm{out}$), reweighting the rows and the columns of $\widehat{\mathbf{T}}$ such that the row marginals and the column marginals of the new rescaled matrix are equal to $\bm{\sigma}_\mathrm{in}$ and $\bm{\sigma}_\mathrm{out}$, just like in the context of standard optimal transport with entropy regularization \cite{Cuturi2013}.

Indeed, this procedure is closely related to the solution of the standard, relaxed, optimal transport problem with entropy regularization when using the matrix $\widehat{\mathbf{T}}$ as the matrix containing the costs or rewards of transportation \cite{Wilson-1970,Kapur-1989,Erlander-1990,Kapur-1992,Cuturi2013}.
It is usually solved by iterative proportional fitting, matrix balancing or biproportional scaling (\cite{Bacharach-1965,Sinkhorn-1967}; see \cite{Kurras-2015,Pukelsheim-2014} and the references therein for a more recent discussion). The iterative proportional fitting algorithm has guaranteed convergence to a unique solution under some mild conditions (see, e.g., \cite{Kurras-2015,Pukelsheim-2014}).


Our iterative procedure for solving Equations (\ref{it1}-\ref{it2}) consists of first fixing an arbitrary $ \bm{\mu}_\mathrm{out}^{(0)}$, and then computing $\bm{\mu}_\mathrm{in}^{(\tau)}$ from (\ref{it1}) at iteration $\tau$. Thereafter, $\bm{\mu}_\mathrm{in}^{(\tau)}$ is kept fixed and $\bm{\mu}_\mathrm{out}^{(\tau + 1)}$ is computed from $\bm{\mu}_\mathrm{in}^{(\tau)}$ in (\ref{it2}).

Finally, note that the optimal probability distribution $\mathrm{P}^\star(\wp)$ given by Proposition \ref{prop::optimal_sol} is useless in practice, as there is an infinite number of paths. However, the different interesting and useful quantities can also be computed from the fundamental matrix and the Lagrange multipliers, as shown in the next section. Moreover, it will be shown in Subsection \ref{Subsec_regularized_optimal_transport01} that the optimization problem can be reduced to the estimation of the joint probabilities $\mathrm{P}^\star(S = i, E = j) = \sum_{\wp_{ij} \in \mathcal{P}_{ij}} \mathrm{P}^\star(\wp_{ij})$ (the coupling), therefore completely avoiding the introduction of the probability distribution. 



\subsection{Computation of other important quantities}

Similarly to \cite{franccoisse2016bag}, other meaningful quantities can be computed in closed form after the convergence of $\bm{\mu}_\mathrm{in}$ and $\bm{\mu}_\mathrm{out}$. This section provides expressions for computing them. The computation of the Lagrange multipliers together with the derived quantities, for the regular, non-hitting, BoP model, are summarized in Algorithm \ref{Alg_nh_bop}.


\subsubsection{The coupling matrix}
\label{Subsec_coupling_matrix_regular_paths01}
First, the definition and computation of the \emph{coupling matrix} is presented. Its name derives from the transportation science literature \cite{villani2003topics,villani2008optimal}.

\begin{proposition} 
\label{prop::coupling}
The coupling matrix, denoted by $\mathbf{\Gamma} = (\gamma_{ij})$ and defined by
\begin{equation}
\gamma_{ij} \triangleq \mathrm{P}^\star(S = i, E = j) = \sum_{\wp_{ij} \in \mathcal{P}_{ij}} \mathrm{P}^\star(\wp_{ij}),
\label{Eq_coupling_probabilities01}
\end{equation}
where $\mathcal{P}_{ij}$ is the set of paths starting in node $i$ and ending in node $j$, can be computed by
\begin{equation}
\mathbf{\Gamma} = \mathbf{Diag}(\bm{\mu}_\mathrm{in} \circ \bm{\sigma}_\mathrm{in}) \; \mathbf{\widehat{Z}} \; \mathbf{Diag}(\bm{\mu}_\mathrm{out} \circ \bm{\alpha}). 
\label{sol_pi_mat}
\end{equation}
\end{proposition}

\begin{proof}
From the optimal probabilities (\ref{prob_ref_comp}) and the definition of the fundamental matrix (\ref{fundamental_matrix}), the probability of drawing a path starting in $i$ and ending $j$ is
\begin{align}
\mathrm{P}^\star(S = i, E = j) &= \sum_{\wp_{ij} \in \mathcal{P}_{ij}} \mathrm{P}^\star(\wp_{ij}) \nonumber  \\
&= \mu^\mathrm{in}_i \sigma^\mathrm{in}_{i} \mu^\mathrm{out}_j  \alpha_j \sum_{\wp_{ij} \in \mathcal{P}_{ij}} \widehat{\pi}^\mathrm{ref}(\wp_{ij}) \exp( - \beta \tilde{c}(\wp_{ij})) \\
&= \mu^\mathrm{in}_i \sigma^\mathrm{in}_{i} \widehat{z}_{ij} \mu^\mathrm{out}_j  \alpha_j,
\label{pi_element}
\end{align}
which in matrix form writes out as Equation (\ref{sol_pi_mat}).
\end{proof}


\subsubsection{The optimal free energy}

\begin{proposition} 
\label{prop::optimal_FE}
The value of the free energy at the optimal paths probability distribution $\mathrm{P}^\star$ is
\begin{equation}
\mathrm{FE}(\mathrm{P}^\star) = - \left( \bm{\lambda}_\mathrm{in}^\top \bm{\sigma}_\mathrm{in} + \bm{\lambda}_\mathrm{out}^\top \bm{\sigma}_\mathrm{out} \right),
\label{FE}
\end{equation}
where the Lagrange parameter vectors $\bm{\lambda}_\mathrm{in}$, $\bm{\lambda}_\mathrm{in}$ are obtained from the vectors $\bm{\mu}_\mathrm{in}$, $\bm{\mu}_\mathrm{in}$ as stated in Proposition \ref{prop::optimal_sol}.
\end{proposition}

\begin{proof}
By replacing $\mathrm{P}^\star(\wp_{ij}) = \mu^\mathrm{in}_{i} \mu^\mathrm{out}_{j} \mathrm{P}^\mathrm{ref}(\wp_{ij}) \exp(-\beta \tilde{c}(\wp_{ij}))$ into the free energy expression (\ref{init_prob}), we get
\begin{gather*}
\text{FE}(\mathrm{P}^\star) = -\sum_{i \in \mathcal{V}} \lambda^\mathrm{in}_{i}  \sum_{j \in \mathcal{V}} \sum_{\wp_{ij} \in \mathcal{P}_{ij}} \mathrm{P}^\star(\wp_{ij}) - \sum_{j \in \mathcal{V}} \lambda^\mathrm{out}_{j}  \sum_{i \in \mathcal{V}} \sum_{\wp_{ij} \in \mathcal{P}_{ij}} \mathrm{P}^\star(\wp_{ij})
\end{gather*}
and, as the margins are fixed, we get the result.
\end{proof}
An interesting interpretation of this proposition based on an optimal transportation analogy is discussed in Section \ref{sec::transp}.


\subsubsection{The expected number of passages through an edge}

Let us define the matrix $\mathbf{\bar{N}} = (\bar{n}_{ij})$ as the matrix containing the expected number of times an edge $(i,j)$ appears on a path drawn from the optimal distribution $\mathrm{P}^\star(\wp)$. Formally,
\begin{equation*}
\bar{n}_{ij}  \triangleq \sum_{k,l \in \mathcal{V}} \sum_{\wp_{kl} \in \mathcal{P}_{kl}} \mathrm{P}^\star(\wp_{kl}) \eta(i \to j \in \wp_{kl}),
\end{equation*}
where $\eta(i \to j \in \wp_{kl})$ denotes the number of times the edge $(i,j)$ appears on path $\wp_{kl}$. 

In \cite{guex2016interpolating}, $\bar{n}_{ij}$ is interpreted as the \emph{flow on edges}, creating a ``stream of matter" going from supply nodes in $\mathcal{I}n$ to destination nodes in $\mathcal{O}ut$. This interpretation will also be discussed further in Section \ref{sec::transp}.

\begin{proposition}
The matrix $\mathbf{\bar{N}}$, containing the expected number of times an edge $(i,j)$ appears on a drawn path, is given by 
\begin{equation}
\mathbf{\bar{N}}= \mathbf{Diag}(\mathbf{\bar{n}}^\mathrm{ref} \div \bm{\mu}_\mathrm{out}) \; \widehat{\mathbf{W}} \; ( \mathbf{Diag}( \bm{\mu}_\mathrm{in} ) )^{-1},
\end{equation}
where $\widehat{\mathbf{W}}$ is defined after Equation (\ref{fundamental_matrix}).
\end{proposition}

\begin{proof}
Because costs are additive along paths, $\eta(i \to j \in \wp_{kl}) = \partial \tilde{c}(\wp_{kl}) / \partial c_{ij}$. Then, by using (\ref{prob_ref_comp}), we obtain
\begin{align}
\bar{n}_{ij} &= \sum_{k,l \in \mathcal{V}} \mu^\mathrm{in}_k  \sigma^\mathrm{in}_{k}  \mu^\mathrm{out}_l  \alpha_{l} \sum_{\wp_{kl} \in \mathcal{P}_{kl}} \widehat{\pi}^\mathrm{ref}(\wp_{kl})  \exp(-\beta \tilde{c}(\wp_{kl})) \frac{ \partial \tilde{c}(\wp_{kl})}{ \partial c_{ij}} \notag \\
&= - \frac{1}{\beta} \sum_{k,l \in \mathcal{V}}  \mu^\mathrm{in}_k \sigma^\mathrm{in}_{k} \mu^\mathrm{out}_l  \alpha_l \frac{ \partial  \widehat{z}_{kl} }{\partial c_{ij}}. \label{res_nij}
\end{align}
We know from direct calculus \cite[Equation (11)]{kivimaki2016two} that $ \partial  \widehat{z}_{kl} / \partial c_{ij} = - \beta \widehat{w}_{ij} \widehat{z}_{ki} \widehat{z}_{jl}$. Thus,
\begin{equation}
\bar{n}_{ij} = \widehat{w}_{ij} \left( \sum_{k \in \mathcal{V}}  \mu^\mathrm{in}_k  \sigma^\mathrm{in}_{k} \widehat{z}_{ki} \right) \left( \sum_{l \in \mathcal{V}}  \mu^\mathrm{out}_l \alpha_l \widehat{z}_{jl} \right),\label{pass_ij_elem_h}
\end{equation}
and by using (\ref{it1_elem})-(\ref{it2_elem}) with $\alpha_l  = \sigma^\mathrm{out}_l / \bar{n}^\mathrm{ref}_l$ (see Proposition \ref{prop::piout_to_alpha}),
\begin{equation}
\bar{n}_{ij} = \frac{\bar{n}^\mathrm{ref}_i \widehat{w}_{ij}}{\mu^\mathrm{out}_i \mu^\mathrm{in}_j}, \label{pass_ij_elem_s}
\end{equation}
which provides the desired result.
\end{proof}

\begin{algorithm}[t!]
\caption[Non-hitting bag-of-paths algorithm]
{Regular, non-hitting, margin-constrained bag-of-paths algorithm}
\algsetup{indent=2em, linenodelimiter=.}
\begin{algorithmic}[1]
\small
\REQUIRE $\,$ \\
 -- The $n\times n$ non-negative adjacency matrix $\mathbf{A}$ of a strongly connected directed graph, containing edge affinities.\\
 -- The $n\times n$ non-negative cost matrix $\mathbf{C}$, containing edge costs.\\
 -- An input distribution vector $\bm{\sigma}_\mathrm{in}$ of size $n$.\\
 -- An output distribution vector $\bm{\sigma}_\mathrm{out}$ of size $n$.\\
 -- A killing rate vector $\bm{\alpha}$ of size $n$.\\
 -- The expected number of visits to each node for the reference random walk, $\mathbf{\bar{n}^\mathrm{ref}}$, of size $n$.\\
 -- A scalar inverse temperature parameter $\beta > 0$.
 
\ENSURE $\,$ \\
 -- The coupling matrix $\bm{\Gamma}$ of size $n\times n$. \\
 -- The minimum free energy $\text{FE}^\mathrm{min}$ (a scalar). \\ 
 -- The matrix $\mathbf{\bar{N}}$ of size $n\times n$, containing the expected number of visits to each edge $(i,j)$. \\
 -- The vector $\bar{\mathbf{n}}$ of size $n$, containing the expected number of visits to each node. \\
 -- The resulting biased transition probabilities matrix $\mathbf{P}$ of size $n\times n$ (the policy). \\
~\\

\STATE $\mathbf{D} \leftarrow \mathbf{Diag}(\mathbf{A}\mathbf{e})$ \COMMENT{the out-degree matrix} \\
\STATE $\mathbf{P}^\mathrm{ref} \leftarrow \mathbf{D}^{-1} \mathbf{A}$ \COMMENT{the reference transition probabilities matrix} \\
\STATE $\mathbf{\widehat{P}}^\mathrm{ref} \leftarrow (\mathbf{I} - \mathbf{Diag}(\bm{\alpha}))\mathbf{P}^\mathrm{ref}$ \COMMENT{the killed reference transition probabilities matrix} \\
\STATE $\widehat{\mathbf{W}} \leftarrow \mathbf{\widehat{P}}^\mathrm{ref} \circ \exp[-\beta \mathbf{C}]$ \COMMENT{the auxiliary matrix $\widehat{\mathbf{W}}$}; $\circ$ is the elementwise product \\
\STATE $\mathbf{\widehat{Z}} \leftarrow (\mathbf{I} - \widehat{\mathbf{W}})^{-1}$ \COMMENT{the fundamental matrix} \\
\STATE $\bm{\mu}_\mathrm{out} \leftarrow \mathbf{e}$ \COMMENT{the initial Lagrange parameters vector for ouputs} \\

\REPEAT[main iteration loop]
\STATE $\bm{\mu}_\mathrm{in} \leftarrow \mathbf{e} \div  \left( \mathbf{\widehat{Z}}\left( \bm{\mu}_\mathrm{out} \circ \bm{\alpha} \right) \right)$ \COMMENT{the update of input vector; $\div$ is the elementwise division} \\ 
\STATE $\bm{\mu}_\mathrm{out} \leftarrow \bar{\mathbf{n}}^\mathrm{ref} \div  \left( \mathbf{\widehat{Z}}^\top \left( \bm{\mu}_\mathrm{in} \circ \bm{\sigma}_\mathrm{in} \right) \right)$ \COMMENT{the update of output vector} \\ 
\UNTIL{convergence}

\STATE $\mathbf{\Gamma} \leftarrow \mathbf{Diag}(\bm{\mu}_\mathrm{in} \circ \bm{\sigma}_\mathrm{in}) \; \mathbf{\widehat{Z}} \; \mathbf{Diag}(\bm{\mu}_\mathrm{out} \circ \bm{\alpha})$ \COMMENT{the coupling matrix} \\
\STATE $\text{FE}^\mathrm{min} \leftarrow \frac{1}{\beta} ( (\log[\bm{\mu}_\mathrm{in}])^\top \bm{\sigma}_\mathrm{in} + (\log[\bm{\mu}_\mathrm{out}])^\top \bm{\sigma}_\mathrm{out})$ \COMMENT{the minimum free energy value} \\
\STATE $\mathbf{\bar{N}} \leftarrow \mathbf{Diag}(\mathbf{\bar{n}}^\mathrm{ref} \div \bm{\mu}_\mathrm{out}) \; \widehat{\mathbf{W}} \; (\mathbf{Diag}( \bm{\mu}_\mathrm{in} ))^{-1}$ \COMMENT{the matrix containing the expected number of visits to each edge} \\
\STATE $\bar{\mathbf{n}} \leftarrow \bar{\mathbf{n}}^\mathrm{ref} \div ( \bm{\mu}_\mathrm{in} \circ \bm{\mu}_\mathrm{out} )$ \COMMENT{the vector containing the expected number of visits to each node} \\
\STATE $\mathbf{P} \leftarrow \mathrm{pinv}(\mathbf{Diag}(\mathbf{\bar{N}} \mathbf{e})) \, \mathbf{\bar{N}}$ \COMMENT{the biased transition matrix (the policy)} \\

\RETURN $\mathbf{\Gamma}$, $\text{FE}^\mathrm{min}$, $\mathbf{\bar{N}}$, $\bar{\mathbf{n}}$, $\mathbf{P}$
\end{algorithmic} \label{Alg_nh_bop} 
\end{algorithm}


\subsubsection{The expected number of visits to a node}
\label{sec::nnodes}

Let us further define the vector $\bar{\mathbf{n}} = (\bar{n}_i)$, containing the expected number of times node $i$ is drawn under $\mathrm{P}^\star(\wp)$, by
\begin{equation*}
\bar{n}_i \triangleq \sum_{\wp \in \mathcal{P}} \mathrm{P}^\star(\wp) \, \eta(i \in \wp),
\end{equation*}
where $\eta(i \in \wp)$ denotes the number of times node $i$ appears on path $\wp$. 

\begin{proposition}
The vector  $\bar{\mathbf{n}}$, containing  the expected number of times node $i$ is drawn, is 
\begin{equation}
\bar{\mathbf{n}} = \bar{\mathbf{n}}^\mathrm{ref} \div ( \bm{\mu}_\mathrm{in} \circ \bm{\mu}_\mathrm{out} ),
\end{equation}
where $\div$ and $\circ$ are the elementwise  division and product.
\end{proposition}

\begin{proof}
In fact, we can decompose this quantity as 
\begin{equation}
\eta(i \in \wp) = \sum_{j \in \mathcal{V}} \eta(i \to j \in \wp) + \delta(E(\wp) = i), \label{eta_dev}
\end{equation}
where $\delta(E(\wp) = i)$ is one if $i$ is the ending node of the path $\wp$. This gives
\begin{align}
\bar{n}_i &= \sum_{j \in \mathcal{V}} \sum_{\wp \in \mathcal{P}} \mathrm{P}^\star(\wp) \eta(i \to j \in \wp) + \sum_{\wp \in \mathcal{P}} \mathrm{P}^\star(\wp) \delta(E(\wp) = i) = \sum_{j \in \mathcal{V}} \bar{n}_{ij} + \sigma^\mathrm{out}_i.
\label{n_i_sum}
\end{align}
By using (\ref{pass_ij_elem_h}) we find
\begin{equation*}
\bar{n}_i = \left( \sum_{k \in \mathcal{V}}  \mu^\mathrm{in}_k  \sigma^\mathrm{in}_{k} \widehat{z}_{ki} \right) \left( \sum_{l \in \mathcal{V}}  \mu^\mathrm{out}_l \alpha_l \sum_{j \in \mathcal{V}} \widehat{w}_{ij}\widehat{z}_{jl} \right) + \sigma^\mathrm{out}_i.
\end{equation*}
However, because $(\mathbf{I} - \widehat{\mathbf{W}}) \widehat{\mathbf{Z}} = \mathbf{I}$, $\sum_{j \in \mathcal{V}} \widehat{w}_{ij} \widehat{z}_{jl}  = \widehat{z}_{il} - \delta_{il}$ as shown in \cite[Equation (13)]{kivimaki2016two}; therefore
\begin{equation*}
\bar{n}_i = \sum_{k \in \mathcal{V}}  \mu^\mathrm{in}_k  \sigma^\mathrm{in}_{k} \widehat{z}_{ki}  \left(  \sum_{l \in \mathcal{V}}  \mu^\mathrm{out}_l \alpha_l \widehat{z}_{il} - \mu^\mathrm{out}_i \alpha_i \right) + \sigma^\mathrm{out}_i,
\end{equation*} 
which, by using (\ref{it1_elem})-(\ref{it2_elem}) and $\alpha_i = \sigma^\mathrm{out}_i / \bar{n}^\mathrm{ref}_i$  (see Proposition \ref{prop::piout_to_alpha}), provides the result
\begin{equation}
\bar{n}_i =  \frac{\bar{n}^\mathrm{ref}_i}{ \mu^\mathrm{in}_i \mu^\mathrm{out}_i }. \label{pass_i_elem_s}
\end{equation}
\end{proof}


\subsubsection{The optimal randomized policy}

The expected number of times edges are visited induces a \emph{biased} random walk on the network with transition matrix $\mathbf{P} \triangleq (p_{ij})$ provided by
\begin{equation}
p_{ij} \triangleq \frac{\bar{n}_{ij} }{\sum_{k \in \mathcal{V}} \bar{n}_{ik}}.
\label{eq:BiasedProbabilities_definition}
\end{equation}

\begin{proposition} 
The biased random walk transition matrix, $\mathbf{P}$, called the randomized routing policy, is provided by
\begin{equation}
\label{eq:BiasedProbabilities}
\mathbf{P} =  \widehat{\mathbf{W}}  \div \left( \left( (\mathbf{e} \div \bm{\mu}_\mathrm{in}) - (\bm{\mu}_\mathrm{out} \circ \bm{\alpha}) \right) \bm{\mu}_\mathrm{in}^\top \right).
\end{equation}
\end{proposition}

\begin{proof}
The result is obtained in elementwise form by combining expression (\ref{eq:BiasedProbabilities_definition}) with (\ref{pass_ij_elem_s}), (\ref{n_i_sum}) and (\ref{pass_i_elem_s}) and using $\alpha_i = \sigma^\mathrm{out}_i / \bar{n}^\mathrm{ref}_i$ (see Proposition \ref{prop::piout_to_alpha}):
\begin{equation*}
p_{ij} = \frac{\mu^\mathrm{in}_i \widehat{w}_{ij}}{\mu^\mathrm{in}_j (1 - \mu^\mathrm{in}_i \mu^\mathrm{out}_i \alpha_i)}.
\end{equation*}
\end{proof}
We can observe that when $\beta \to 0$, then $\bm{\mu}_\mathrm{in} \to \mathbf{e}$, $\bm{\mu}_\mathrm{out} \to \mathbf{e}$, $\widehat{\mathbf{W}} \to \mathbf{\widehat{P}}^\mathrm{ref}$ and we obtain
\begin{equation*}
p_{ij} \xrightarrow[T \to \infty]{} \frac{\widehat{p}^\mathrm{ref}_{ij}}{(1 - \alpha_i)} = p_{ij}^\mathrm{ref},
\end{equation*}
as it should be.

This biased random walk is the \emph{optimal policy} that has to be followed for reaching nodes in $\mathcal{O}ut$ from nodes in $\mathcal{I}n$, and can be interpreted as follows. When $\beta \to 0$, the behavior becomes similar to the random walk defined by $\mathbf{P}^\mathrm{ref}$, but as $\beta$ increases, random walkers are more and more ``attracted" by high $\sigma^\mathrm{out}_i$ nodes. These ``pools of attraction", whose sizes are related to the components of $\bm{\sigma}_\mathrm{out}$, get less and less ``fuzzy" as $\beta$ increases, eventually forcing walkers to adopt quasi-deterministic, optimal, paths following the solution of the optimal transport on a graph problem (see Subsection \ref{sec::speed}). Therefore, this framework can be viewed as an extension of standard electrical networks, interpolating between an optimal behavior based on least cost paths and a random behavior based on the reference probabilities \cite{guex2016interpolating}.


\section{The margin-constrained bag-of-hitting-paths formalism}
\label{Sec_formalism_hitting_paths01}

In \cite{franccoisse2016bag}, the BoP formalism was defined for regular paths (as in previous section) as well as for hitting paths, i.e.\ paths where the final node $j$ appears only once as last node of the path. In this section, we will now consider the margin-constrained problem for \emph{hitting paths} and define, accordingly, the \emph{margin-constrained bag-of-hitting-paths framework} (abbreviated as cBoHP). This new model will yield interesting properties and analogies with other models, and will require less computation time in comparison to the non-hitting, regular, bag-of-paths model considered so far in this work. We will see that while both models are similar when $T \to 0$, they are quite different when $T \to \infty$, and that the hitting formalism has a somewhat more straightforward solution. Nevertheless, the hitting paths assumption can prove more practical and appropriate in practice. But, of course, the choice of whether to consider non-hitting or hitting paths depends on the application.


\subsection{Problem definition}

Let  $\mathcal{P}^\mathrm{h}_{ij}$ be the set of all \emph{hitting paths} starting in $i$ and ending in $j$, i.e.\ all paths $\wp_{ij} = (i_0,\dots,i_t)$ where $i_0 = i$, $i_t = j$ and $i_\tau \neq j, \forall \tau \neq t$ where $t$ is the length of the path. This means that, technically, the ending node is turned into a killing, absorbing, node from which we cannot escape \cite{franccoisse2016bag,kivimaki2014developments,fouss2016algorithms}. We define the set of all hitting paths, also named \emph{bag-of-hitting-paths}, by $\mathcal{P}_\mathrm{h} = \cup_{i,j \in \mathcal{V}} \mathcal{P}^\mathrm{h}_{ij}$. By analogy with (\ref{init_prob}), the problem here is to find the \emph{optimal hitting paths probability distribution}, $\mathrm{P}^\star_\mathrm{h}(\wp)$, solving
\begin{equation}
\label{init_prob_h}
\vline \begin{array}{ll@{}ll}
\underset{\{ \mathrm{P}_\mathrm{h}(\wp) \}_{\wp \in \mathcal{P}_\mathrm{h}}}{\text{minimize}} & \text{FE}_\mathrm{h}(\mathrm{P}_\mathrm{h}) =  \displaystyle\sum\limits_{\wp \in \mathcal{P}_\mathrm{h}} \mathrm{P}_\mathrm{h}(\wp) \tilde{c}(\wp) + T \sum_{\wp \in \mathcal{P}_\mathrm{h}} \mathrm{P}_\mathrm{h}(\wp) \log \left( \frac{\mathrm{P}_\mathrm{h}(\wp)}{\mathrm{P_\mathrm{h}^\mathrm{ref}}(\wp)} \right) \\
\text{subject to} & \sum_{j \in \mathcal{V}} \sum_{\wp_{ij} \in \mathcal{P}^\mathrm{h}_{ij}} \mathrm{P}_\mathrm{h}(\wp_{ij}) = \sigma^\mathrm{in}_i \qquad \forall i \in \mathcal{V},\\
          & \sum_{i \in \mathcal{V}} \sum_{\wp_{ij} \in \mathcal{P}^\mathrm{h}_{ij}} \mathrm{P}_\mathrm{h}(\wp_{ij}) = \sigma^\mathrm{out}_j \qquad \forall j \in \mathcal{V}. \\ 
\end{array}
\end{equation}
As probabilities for regular paths containing the final node more than once converge to zero in the non-hitting formalism when $T \to 0$ (it is sub-optimal to visit several times the same node), we easily see that both problems are equivalent at this limit. However, this is not the case when $T \to \infty$, due to the difference in reference probabilities between the two models and the structure of the paths, as shown in the next section.

Yet another important difference between the hitting and the non-hitting formulations is that the former model \emph{is equivalent to the standard entropy regularized optimal transport problem} \cite{Wilson-1970,Erlander-1990,Cuturi2013} based on the directed free energy distance (or potential) between nodes \cite{kivimaki2014developments,franccoisse2016bag,fouss2016algorithms}, which can easily be pre-computed for the whole graph. This is detailed in Section \ref{sec::transp} (see Equation (\ref{altern_prob_hitting})). A last difference is that the hitting paths formulation does not need the pre-processing step computing the reference probabilities with fixed margins described in Section
\ref{Sec_non_hitting_reference_probabilities01} and Algorithm \ref{Alg_kill_prob}.


\subsection{Reference probabilities with fixed margins}

Finding reference probabilities $\mathrm{P}_\mathrm{h}^\mathrm{ref}(\wp)$ is at the heart of the difference between both formalisms, and is greatly simplified in the hitting case. In fact, it was shown in \cite{franccoisse2016bag} that the sum of likelihoods over all hitting paths between two nodes $i$ and $j$ is always equal to $1$. In other words,
\begin{equation*}
\sum_{\wp \in \mathcal{P}^\mathrm{h}_{ij}} \tilde{\pi}^\mathrm{ref}(\wp) = 1, 
\end{equation*}
where $\tilde{\pi}^\mathrm{ref}(\wp) \triangleq \Pi_{\tau = 1}^t p^\mathrm{ref}_{i_{\tau-1},i_{\tau}}$ and with the $p^\mathrm{ref}_{ij}$ being the reference transition probabilities (\ref{Eq_define_transition_probabilities_elementwise01}). From this, it is easy to observe that the reference probability defined by
\begin{equation}
\mathrm{P}^\mathrm{ref}_\mathrm{h}(\wp_{ij}) \triangleq \sigma^\mathrm{in}_i \sigma^\mathrm{out}_j \tilde{\pi}^\mathrm{ref}(\wp_{ij}), 
\label{prob_ref_h}
\end{equation}
yields the correct margins as expressed in the constraints of problem (\ref{init_prob_h}).


\subsection{Computation of the optimal probability distribution over paths}

The reasoning for finding the solution follows the same rationale as before for the non-hitting case; therefore the details of the proofs are not repeated in this section. The main difference lies in the replacement of the reference probabilities with (\ref{prob_ref_h}), and the following new expression for hitting paths (the equivalent of Equation (\ref{Eq_fundamental_matrix_regular_paths01}))
\begin{equation}
\sum_{\wp_{ij} \in \mathcal{P}^\mathrm{h}_{ij}} \tilde{\pi}^\mathrm{ref}(\wp_{ij}) \exp( - \beta \tilde{c}(\wp_{ij})) = z^\mathrm{h}_{ij},
\end{equation}
with $\mathbf{Z}_\mathrm{h} = (z^\mathrm{h}_{ij})$ being the \emph{fundamental matrix for hitting paths} as introduced in (\cite{kivimaki2014developments}, Equation (12)), obtained through
\begin{equation}
\mathbf{Z}_\mathrm{h} \triangleq  \mathbf{Z} \mathbf{D}_\mathrm{h}^{-1},
\end{equation}
where $\mathbf{Z} \triangleq (\mathbf{I} - \mathbf{W})^{-1}$, $\mathbf{W} \triangleq \mathbf{P}^\mathrm{ref} \circ \exp[-\beta \mathbf{C}]$, and $\mathbf{D}_\mathrm{h}\triangleq \mathbf{Diag}(\mathbf{Z})$, the diagonal matrix containing the main diagonal of $\mathbf{Z}$. Elementwise, we have $z^\mathrm{h}_{ij} = z_{ij} / z_{jj}$.

\subsubsection{The optimal hitting-paths probabilities}

The optimal hitting-paths probabilities are obtained with the following proposition.

\begin{proposition}
\label{prop::optimal_sol_h}
When the set of paths $\mathcal{P}$ is restricted to the set of hitting paths $\mathcal{P}_{\mathrm{h}}$ and the reference path probabilities $\mathrm{P}^{\mathrm{ref}}(\wp)$ are defined according to Equation $(\ref{prob_ref_h})$, then the minimization problem $(\ref{init_prob_h})$ is solved by
\begin{equation}
\mathrm{P}^\star_\mathrm{h}(\wp_{ij}) = \mu^\mathrm{h,in}_{i}  \sigma^\mathrm{in}_{i} \mu^\mathrm{h,out}_{j} \sigma^\mathrm{out}_{j}  \tilde{\pi}^\mathrm{ref}(\wp_{ij}) \exp(-\beta \tilde{c}(\wp_{ij} )),\label{prob_ref_comp_h} \\
\end{equation}
where $\beta \triangleq 1 / T$ is the inverse temperature parameter, and $\mu^\mathrm{h,in}_i \triangleq \exp( - \beta \lambda^\mathrm{h,in}_i)$, $\mu^\mathrm{h,out}_j \triangleq \exp( - \beta \lambda^\mathrm{h,out}_j)$ are two vectors derived from the Lagrange parameter vectors $\bm{\lambda}_\mathrm{h,in}$, $\bm{\lambda}_\mathrm{h,out}$, associated with the constraints.
\end{proposition}

\begin{proof}
The proof is similar to the proof of Proposition \ref{prop::optimal_sol}.
\end{proof}

\subsubsection{Computing the Lagrange parameters}

\begin{proposition} 
\label{prop::dual_vectors_h}
The two vectors $\bm{\mu}_\mathrm{h,in}$ and $\bm{\mu}_\mathrm{h,out}$, defined by $\mu^\mathrm{h,in}_i = \exp( - \beta \lambda^\mathrm{h,in}_i)$ and $\mu^\mathrm{h,out}_i = \exp( - \beta \lambda^\mathrm{h,out}_i)$, verify the following expressions
\begin{eqnarray}
\bm{\mu}^\mathrm{h}_\mathrm{in} = \mathbf{e} \div \left( \mathbf{Z}_\mathrm{h} \left( \bm{\mu}^\mathrm{h}_\mathrm{out} \circ \bm{\sigma}_\mathrm{out} \right) \right), \label{it1_h} \\
\bm{\mu}^\mathrm{h}_\mathrm{out} = \mathbf{e} \div \left( \mathbf{Z}_\mathrm{h}^\top \left( \bm{\mu}^\mathrm{h}_\mathrm{in} \circ \bm{\sigma}_\mathrm{in} \right) \right).
\label{it2_h}
\end{eqnarray} 
\end{proposition} 
\begin{proof}
The proof is similar to the proof of Proposition \ref{prop::dual_vectors}, but by replacing the matrix $\mathbf{Z}$ with $\mathbf{Z}_\mathrm{h}$.
\end{proof}

As for regular paths, these two expressions are recomputed iteratively until convergence. Note that in the hitting paths case, the matrix that needs to be rescaled in order to satisfy the margin constraints is
\begin{equation}
\mathbf{T}_{\mathrm{h}} = \mathbf{Diag}(\bm{\sigma}_\mathrm{in}) \mathbf{Z}_{\mathrm{h}} \mathbf{Diag}(\bm{\sigma}_\mathrm{out}).
\label{Eq_T_matrix_hitting01}
\end{equation}

The computation of the Lagrange multipliers and the derived quantities for the hitting bag-of-paths model is summarized in Algorithm \ref{Alg_h_bop}.


\subsection{Computation of other important quantities}

The computation of the other interesting quantities is slightly different in the hitting formalism, and shows some interesting new properties. They are reviewed in this section.


\subsubsection{The coupling matrix}
\label{Subsec_coupling_matrix_hitting_paths01}

\begin{proposition} 
\label{prop::coupling_h}
The coupling matrix for hitting paths $\mathbf{\Gamma}_\mathrm{h} = (\gamma^\mathrm{h}_{ij})$ is given by 
\begin{equation}
\mathbf{\Gamma}_\mathrm{h} = \mathbf{Diag}(\bm{\mu}^\mathrm{h}_\mathrm{in} \circ \bm{\sigma}_\mathrm{in}) \; \mathbf{Z}_\mathrm{h} \; \mathbf{Diag}(\bm{\mu}^\mathrm{h}_\mathrm{out} \circ \bm{\sigma}_\mathrm{out}). \label{pi_h}
\end{equation}
\end{proposition}

\begin{proof}
The proof is similar to the proof of Proposition \ref{prop::coupling}.
\end{proof}

Notice that, for the hitting paths case, an alternative way of obtaining the coupling matrix by pre-computing the directed free energy distances is discussed in Subsection \ref{Subsec_regularized_optimal_transport01} (see Equation (\ref{altern_prob_hitting})). 


\subsubsection{The optimal free energy}

\begin{proposition} 
\label{prop::optimal_FE_h}
The value of the free energy at the optimal distribution $\mathrm{P}_{\mathrm{h}}^\star$ is
\begin{equation}
\mathrm{FE}_\mathrm{h}(\mathrm{P}^\star_\mathrm{h}) = - \left(  (\bm{\lambda}_\mathrm{in}^\mathrm{h})^\top \bm{\sigma}_\mathrm{in} + (\bm{\lambda}_\mathrm{out}^\mathrm{h})^\top \bm{\sigma}_\mathrm{out} \right).
\label{FE_h}
\end{equation}
\end{proposition}

\begin{proof}
The proof is similar to the proof of Proposition \ref{prop::optimal_FE}.
\end{proof}


\subsubsection{The expected number of visits to an edge}

\begin{proposition}
The matrix $\mathbf{\bar{N}_\mathrm{h}}$, containing the expected number of times an edge $(i,j)$ appears on a hitting path, is given by 
\begin{equation}
\mathbf{\bar{N}}_\mathrm{h} = \left( (\mathbf{e}  \div \bm{\mu}^\mathrm{h}_\mathrm{out}) (\mathbf{e} \div \bm{\mu}^\mathrm{h}_\mathrm{in})^\top - (\mathbf{Z}_\mathrm{h} \mathbf{Diag}(\bm{\sigma}_\mathrm{out})\mathbf{Z}_\mathrm{h})^\top \right) \circ (\mathbf{Diag}(\mathbf{Z}) \mathbf{W}),
\end{equation}
where $\mathbf{Diag}(\mathbf{Z})$ is a diagonal matrix containing the main diagonal of $\mathbf{Z}$ and $\circ$ is the elementwise matrix product.
\end{proposition}

\begin{proof}
Let $\bar{n}^\mathrm{h}_{ij}$ be the expected number of times edge $(i,j)$ appears on a path drawn according to $\mathrm{P}_\mathrm{h}(\wp)$, i.e.,
\begin{equation*}
\bar{n}^\mathrm{h}_{ij}  \triangleq \sum_{k,l \in \mathcal{V}} \sum_{\wp_{kl} \in \mathcal{P}^\mathrm{h}_{kl}} \mathrm{P}_\mathrm{h}(\wp_{kl}) \, \eta(i \to j \in \wp_{kl}),
\end{equation*}
where $\eta(i \to j \in \wp_{kl})$ denotes the number of times the edge $(i,j)$ is visited along hitting path $\wp_{kl}$. By a reasoning similar to (\ref{res_nij}) we get 
\begin{equation*}
\bar{n}^\mathrm{h}_{ij} = - \tfrac{1}{\beta} \sum_{k,l \in \mathcal{V}}  \mu^\mathrm{h,in}_k \sigma^\mathrm{in}_{k} \mu^\mathrm{h,out}_l   \sigma^\mathrm{out}_l \frac{ \partial z^\mathrm{h}_{kl} }{\partial c_{ij}}.
\end{equation*}
This time, we have \cite[Equation (11)]{kivimaki2016two},
\begin{equation*}
\frac{ \partial z^\mathrm{h}_{kl} }{ \partial c_{ij} } =    \frac{ \partial \left( z_{kl} / z_{ll} \right) }{ \partial c_{ij}} = - \beta w_{ij} \left( \frac{ z_{ki} z_{jl} }{z_{ll}} - \frac{z_{kl} z_{li} z_{jl} }{z^2_{ll}} \right),
\end{equation*}
leading to
\begin{equation}
\bar{n}^\mathrm{h}_{ij} = \sum_{{k,l \in \mathcal{V}}} \frac{\mu^\mathrm{h,in}_k \sigma^\mathrm{in}_{k} \mu^\mathrm{h,out}_l  \sigma^\mathrm{out}_l }{z_{ll}} \left( {z_{ki}} - \frac{z_{kl} z_{li}}{z_{ll}} \right) w_{ij} z_{jl}. \label{int_n_h_ij}
\end{equation}
From (\ref{pi_h}), and recalling that $z^\mathrm{h}_{kl} = z_{kl}/z_{ll}$, we get $(\mu^\mathrm{h,in}_k \sigma^\mathrm{in}_{k} \mu^\mathrm{h,out}_l  \sigma^\mathrm{out}_l ) / z_{ll} = \gamma^\mathrm{h}_{kl} / z_{kl}$, where $\gamma^\mathrm{h}_{kl}$ is element $k,l$ of the coupling matrix (see Equation (\ref{pi_h})), thus
\begin{equation}
\bar{n}^\mathrm{h}_{ij} =  \sum_{{k,l \in \mathcal{V}}} \gamma^\mathrm{h}_{kl} \left( \frac{z_{ki}}{z_{kl}} - \frac{ z_{li} }{z_{ll}} \right) w_{ij} z_{jl} =  \sum_{{k,l \in \mathcal{V}}} \gamma^\mathrm{h}_{kl} \bar{n}^{(kl)}_{ij}, \label{dec_n_h_ij}
\end{equation}
and $ \bar{n}^{(kl)}_{ij} \triangleq \big( \frac{z_{ki}}{z_{kl}} - \frac{ z_{li} }{z_{ll}} \big) w_{ij} z_{jl}$ is the expected number of times $(i,j)$ is visited when the starting node is fixed to $k$ and the ending node to $l$, as shown in \cite[Equation (12)]{kivimaki2016two}. Within the constrained bag-of-hitting-paths formalism, this quantity is simply the average over all starting and ending nodes, weighted by the coupling probabilities, $\bm{\Gamma}_\mathrm{h}$. In fact, $\bar{n}^{h}_{ij}$ can be seen as a \emph{weighted randomized shortest-paths (RSP) betweenness centrality} for edges, compared to the unweighted RSP betweenness centrality defined in \cite{kivimaki2016two}. More precisely, this quantity provides a weighted \emph{group betweenness} between the two sets of nodes, $\mathcal{I}n$ and $\mathcal{O}ut$.

While the decomposition (\ref{dec_n_h_ij}) of $\bar{n}^\mathrm{h}_{ij}$ provides an interesting analogy, we will use another expression to actually compute this quantity. From (\ref{int_n_h_ij}) with the help of (\ref{it1_h}) and (\ref{it2_h}), we get
\begin{equation}
\bar{n}^\mathrm{h}_{ij} = \left( \frac{1}{\mu^\mathrm{h,out}_i \mu^\mathrm{h,in}_j}  - \sum_{l \in \mathcal{V}} \sigma^\mathrm{out}_l z^\mathrm{h}_{li}z^\mathrm{h}_{jl} \right) z_{ii} w_{ij},
\label{res_n_ij_h}
\end{equation}
which provides the expression for computing $\mathbf{\bar{N}_\mathrm{h}} \triangleq (\bar{n}^\mathrm{h}_{ij})$.
\end{proof}

\begin{algorithm}[t!]
\caption[Bag-of-hitting-paths algorithm]
{Margin-constrained bag-of-hitting-paths algorithm}
\algsetup{indent=2em, linenodelimiter=.}
\begin{algorithmic}[1]
\small
\REQUIRE $\,$ \\
 -- The $n\times n$ non-negative adjacency matrix $\mathbf{A}$ of a strongly connected directed graph, containing edge affinities.\\
 -- The $n\times n$ non-negative cost matrix $\mathbf{C}$, containing edge costs.\\
 -- An input distribution vector $\bm{\sigma}_\mathrm{in}$ of size $n$.\\
 -- An output distribution vector $\bm{\sigma}_\mathrm{out}$ of size $n$.\\
 -- A scalar inverse temperature parameter $\beta > 0$.
 
\ENSURE $\,$ \\
 -- The coupling matrix $\bm{\Gamma}_\mathrm{h}$ of size $n\times n$. \\
 -- The minimum free energy scalar $\text{FE}_\mathrm{h}^\mathrm{min}$. \\ 
 -- The matrix $\mathbf{\bar{N}}_\mathrm{h}$ of size $n\times n$, containing the expected number of visits to each edge $(i,j)$. \\
 -- The vector $\bar{\mathbf{n}}_\mathrm{h}$ of size $n$, containing the expected number of visits to each node. \\
 -- The resulting biased transition probabilities matrix $\mathbf{P}_\mathrm{h}$ of size $n\times n$ (the policy). \\
~\\

\STATE $\mathbf{D} \leftarrow \mathbf{Diag}(\mathbf{A}\mathbf{e})$ \COMMENT{the out-degree matrix} \\
\STATE $\mathbf{P}^\mathrm{ref} \leftarrow \mathbf{D}^{-1} \mathbf{A}$ \COMMENT{the reference transition probabilities matrix} \\
\STATE $\mathbf{W} \leftarrow \mathbf{P}^\mathrm{ref} \circ \exp[-\beta \mathbf{C}]$ \COMMENT{the auxiliary matrix $\mathbf{W}$; $\circ$ is the elementwise product} \\
\STATE $\mathbf{Z} \leftarrow (\mathbf{I} - \mathbf{W})^{-1}$ \COMMENT{the fundamental matrix} \\
\STATE $\mathbf{D}_\mathrm{h} \leftarrow \mathbf{Diag}(\mathbf{Z})$ \COMMENT{the column-normalization matrix for hitting paths probabilities} \\
\STATE $\mathbf{Z}_\mathrm{h} \leftarrow \mathbf{Z} \mathbf{D}^{-1}_\mathrm{h}$ \COMMENT{the column-normalized fundamental matrix} \\
\STATE $\bm{\mu}^\mathrm{h}_\mathrm{out} \leftarrow \mathbf{e}$ \COMMENT{the initial Lagrange parameters vector for ouputs} \\

\REPEAT[main iteration loop]
\STATE $\bm{\mu}^\mathrm{h}_\mathrm{in} \leftarrow \mathbf{e} \div \left( \mathbf{Z}_\mathrm{h} \left( \bm{\mu}^\mathrm{h}_\mathrm{out} \circ \bm{\sigma}_\mathrm{out} \right) \right)$ \COMMENT{the update of input vector} \\ 
\STATE $\bm{\mu}^\mathrm{h}_\mathrm{out} \leftarrow \mathbf{e} \div \left( \mathbf{Z}_\mathrm{h}^\top \left( \bm{\mu}^\mathrm{h}_\mathrm{in} \circ \bm{\sigma}_\mathrm{in} \right) \right)$ \COMMENT{the update of output vector} \\ 
\UNTIL{convergence}

\STATE $\mathbf{\Gamma}_\mathrm{h} \leftarrow \mathbf{Diag}(\bm{\mu}^\mathrm{h}_\mathrm{in} \circ \bm{\sigma}_\mathrm{in}) \; \mathbf{Z}_\mathrm{h} \; \mathbf{Diag}(\bm{\mu}^\mathrm{h}_\mathrm{out} \circ \bm{\sigma}_\mathrm{out})$ \COMMENT{the coupling matrix} \\
\STATE $\text{FE}_\mathrm{h}^\mathrm{min} \leftarrow \frac{1}{\beta} ( (\log[\bm{\mu}^\mathrm{h}_\mathrm{in}])^\top \bm{\sigma}_\mathrm{in} + (\log[\bm{\mu}^\mathrm{h}_\mathrm{out}])^\top \bm{\sigma}_\mathrm{out})$ \COMMENT{the minimum free energy value} \\
\STATE $\mathbf{\bar{N}_\mathrm{h}} \leftarrow \left( (\mathbf{e}  \div \bm{\mu}^\mathrm{h}_\mathrm{out}) (\mathbf{e} \div \bm{\mu}^\mathrm{h}_\mathrm{in})^\top - (\mathbf{Z}_\mathrm{h} \mathbf{Diag}(\bm{\sigma}_\mathrm{out})\mathbf{Z}_\mathrm{h})^\top \right) \circ (\mathbf{Diag}(\mathbf{Z}) \mathbf{W})$ \COMMENT{the matrix containing the expected number of visits to each edge} \\
\STATE $\mathbf{\bar{n}_\mathrm{h}} \leftarrow \mathbf{Diag}(\mathbf{Z}) \left( (\mathbf{e} \div (\bm{\mu}^\mathrm{h}_\mathrm{in} \circ \bm{\mu}^\mathrm{h}_\mathrm{out} )) - \mathbf{diag}(\mathbf{Z}_\mathrm{h} \mathbf{Diag}(\bm{\sigma}_\mathrm{out})\mathbf{Z}_\mathrm{h}) \right) + \bm{\sigma}_\mathrm{out}$ \COMMENT{the vector containing the expected number of visits to each node} \\
\STATE $\mathbf{P}_\mathrm{h} \leftarrow \mathrm{pinv}(\mathbf{Diag}(\mathbf{\bar{N}}_\mathrm{h} \mathbf{e})) \, \mathbf{\bar{N}}_\mathrm{h}$ \COMMENT{the biased transition matrix (the policy) }\\

\RETURN $\mathbf{\Gamma}_\mathrm{h}$, $\text{FE}_\mathrm{h}^\mathrm{min}$, $\mathbf{\bar{N}}_\mathrm{h}$, $\bar{\mathbf{n}}_\mathrm{h}$, $\mathbf{P}_\mathrm{h}$
\end{algorithmic} \label{Alg_h_bop} 
\end{algorithm}


\subsubsection{The expected number of visits to a node}

\begin{proposition}
The vector  $\bar{\mathbf{n}}_\mathrm{h}$, containing  the expected number of times node $i$ appears on a hitting path drawn from a bag of hitting paths, is provided by 
\begin{equation}
\bar{\mathbf{n}}_\mathrm{h} = \mathbf{Diag}(\mathbf{Z}) \left[ (\mathbf{e} \div (\bm{\mu}^\mathrm{h}_\mathrm{in} \circ \bm{\mu}^\mathrm{h}_\mathrm{out} )) - \mathbf{diag}(\mathbf{Z}_\mathrm{h} \mathbf{Diag}(\bm{\sigma}_\mathrm{out})\mathbf{Z}_\mathrm{h}) \right] + \bm{\sigma}_\mathrm{out},
\end{equation}
where $\mathbf{diag}(\mathbf{X})$ is a column vector containing the main diagonal of matrix $\mathbf{X}$. 
\end{proposition}

\begin{proof}
Let us define $\bar{n}^\mathrm{h}_{i}$, the expected number of times node $i$ appears on a path under $\mathrm{P}_\mathrm{h}(\wp)$, by
\begin{equation*}
\bar{n}^\mathrm{h}_{i} \triangleq \sum_{\wp \in \mathcal{P}_\mathrm{h}} \mathrm{P}_\mathrm{h}(\wp) \, \eta(i \in \wp),
\end{equation*}
where $\eta(i \in \wp)$ denotes the number of times node $i$ is visited along hitting path $\wp$. Using (\ref{eta_dev}) and (\ref{n_i_sum}) again, we get
\begin{gather}
\bar{n}^\mathrm{h}_{i} = \sum_{j \in \mathcal{V}} \bar{n}^\mathrm{h}_{ij} + \sigma^\mathrm{out}_i, \label{res_n_i_simple}
\end{gather}
which, after using (\ref{dec_n_h_ij}) and $\sum_{j \in \mathcal{V}} w_{ij} z_{jl} = z_{il} - \delta_{il}$  \cite[Equation (13)]{kivimaki2016two}, results in
\begin{equation*}
\bar{n}^\mathrm{h}_{i} = \sum_{{k,l \in \mathcal{V}}} \gamma^\mathrm{h}_{kl} \bar{n}^{(kl)}_{i} +  \sigma^\mathrm{out}_i,
\end{equation*}
where $\bar{n}^{(kl)}_{i} \triangleq \big( \frac{z_{ki}}{z_{kl}} - \frac{ z_{li} }{z_{ll}} \big) z_{il}$ is the number of times $i$ is visited when starting in $k$ and ending in $l$, as defined in \cite{kivimaki2016two}. Notice that $\bar{n}^{(kl)}_{i} = 0$ when $i=l$ and when $k=l$. Again, this quantity can be seen as a \emph{weighted RSP betweenness centrality} for sets of nodes, by analogy with the unweighted RSP betweenness centrality defined in \cite[Equation (15)]{kivimaki2016two}.

However, for computation, it is more convenient to apply the equivalent of Equation (\ref{res_n_ij_h}). Following the same path as for the derivation of (\ref{res_n_ij_h}), we first use $\sum_{j \in \mathcal{V}} {w}_{ij} {z}_{jl}  = {z}_{il} - \delta_{il}$ with (\ref{res_n_i_simple}) on (\ref{int_n_h_ij}), and then (\ref{it1_h}) and (\ref{it2_h}) in order to get
\begin{equation}
\bar{n}^\mathrm{h}_{i} =\left( \frac{1}{\mu^\mathrm{h,out}_i \mu^\mathrm{h,in}_i}  - \sum_{l \in \mathcal{V}} \sigma^\mathrm{out}_l z^\mathrm{h}_{li}z^\mathrm{h}_{il} \right) z_{ii} + \sigma^\mathrm{out}_i,\label{res_n_i_h}
\end{equation}
which provides the expression for computing the vector $\mathbf{\bar{n}_\mathrm{h}} \triangleq (\bar{n}^\mathrm{h}_{i})$.
\end{proof}


\subsubsection{The optimal randomized policy}

\begin{proposition} 
The biased random walk transition matrix, $\mathbf{P}_\mathrm{h} = (p^\mathrm{h}_{ij})$, that is, the randomized routing policy, is given by
\begin{equation}
p^\mathrm{h}_{ij} \triangleq \frac{\bar{n}^\mathrm{h}_{ij} }{\sum_{k \in \mathcal{V}} \bar{n}^\mathrm{h}_{ik}}
= \frac{ \left( \frac{1}{\mu^\mathrm{h,out}_i \mu^\mathrm{h,in}_j}  - \sum_{l \in \mathcal{V}} \sigma^\mathrm{out}_l z^\mathrm{h}_{li}z^\mathrm{h}_{jl} \right) }{\left( \frac{1}{\mu^\mathrm{h,out}_i \mu^\mathrm{h,in}_i}  - \sum_{l \in \mathcal{V}} \sigma^\mathrm{out}_l z^\mathrm{h}_{li}z^\mathrm{h}_{il} \right)} w_{ij}.
\end{equation}
\end{proposition}

\begin{proof}
We get the result by combining (\ref{res_n_ij_h}), (\ref{res_n_i_simple}) and (\ref{res_n_i_h}).
\end{proof}

Again, we observed experimentally that this quantity converges to the reference transition matrix of the reference random walk when $\beta \to 0$, i.e.,
\begin{equation*}
p^\mathrm{h}_{ij} \xrightarrow[\beta \to 0]{} p_{ij}^\mathrm{ref}.
\end{equation*}
Conversely, when $\beta \to \infty$, the problem becomes an optimal transport on a graph problem (see Subsection \ref{sec::speed}).
Note that, for convenience, in Algorithm \ref{Alg_h_bop}, matrix $\mathbf{P_\mathrm{h}} \triangleq (p^\mathrm{h}_{ij})$ is computed thanks to 
\begin{equation*}
\mathbf{P_\mathrm{h}} = (\mathbf{Diag}(\mathbf{\bar{N}_\mathrm{h}} \mathbf{e}))^{-1} \mathbf{\bar{N}_\mathrm{h}}.
\end{equation*}

Let us now turn to a discussion of the relations between the proposed models and the regularized optimal transport problem.



\section{The regularized optimal transport problem  analogy}
\label{sec::transp}

In this section, we will show that both non-hitting (see Equation (\ref{init_prob})) and hitting (see Equation (\ref{init_prob_h})) problems correspond to two different kinds of \emph{regularization} for the \emph{optimal transport problem} \cite{ahuja1993network,guex2017transportation,kantorovich1942translocation,villani2003topics,villani2008optimal}. It therefore generalizes discrete entropy-regularized optimal transport problems \cite{Wilson-1970,Erlander-1990,Kapur-1992} to a graph structure.

\subsection{The optimal transport problem}
\label{SubSec_optimal_transport_problem01}

General optimal transport is a well-known problem defined, for example, in \cite{kantorovich1942translocation,villani2003topics,villani2008optimal}, and the special case where the space is a graph is easily derived from it \cite{guex2017transportation}. Assume we have a subset of nodes $\mathit{I}n \in \mathcal{V}$, called \emph{sources}, with a supply of a certain quantity of matter, while we observe a demand of the same matter in another subset of nodes, $\mathcal{O}ut \in \mathcal{V}$, called \emph{targets}. We suppose that the overall supply is equal to the overall demand, thus these quantities on nodes can be represented, without loss of generality, by their proportion of the total. In other words, supply and demand are represented respectively by two discrete distribution vector $\bm{\sigma}_\mathrm{in}$ and $\bm{\sigma}_\mathrm{out}$, with $\sum_{i \in \mathcal{V}} \sigma^\mathrm{in}_i = \sum_{i \in \mathcal{V}} \sigma^\mathrm{out}_i = 1$ and $\sigma^\mathrm{in}_i,\sigma^\mathrm{out}_i \geq 0, \; \forall i \in \mathcal{V}$.

The goal of the optimal transport problem is to find an optimal \emph{attribution plan} or optimal \emph{coupling} \cite{villani2003topics,villani2008optimal}, i.e.\ a $n\times n$ matrix $\bm{\Gamma}$, where $\gamma_{ij}$ represents the proportion of matter going from $i$ to $j$, in order to fulfill supply and demand. Optimality here means that the cost of transportation of this attribution plan, i.e. $\sum_{i,j \in \mathcal{V}} \gamma_{ij} \, c_{ij}$ where $c_{ij}$ is the cost of transportation from $i$ to $j$, must be minimal. Altogether, we have
\begin{equation}
\vline \begin{array}{ll@{}ll}
\underset{\bm{\Gamma}}{\text{minimize}} & \sum_{i,j \in \mathcal{V}} \gamma_{ij} \, c_{ij} \\
\text{subject to} & \sum_{j \in \mathcal{V}} \gamma_{ij}= \sigma^\mathrm{in}_i \qquad \; \forall i \in \mathcal{V},\\
          & \sum_{i \in \mathcal{V}} \gamma_{ij} = \sigma^\mathrm{out}_j \qquad \forall j \in \mathcal{V}, \\
          & \gamma_{ij} \geq 0 \qquad \qquad \qquad \forall i,j \in \mathcal{V}.
\end{array} \label{altern_prob}
\end{equation}
Another interesting interpretation can be found in the \emph{dual optimal transport problem} \cite{guex2017transportation,villani2003topics,villani2008optimal}, expressed by
\begin{equation}
\vline \begin{array}{ll@{}ll}
\underset{\bm{\lambda}_\mathrm{in}, \bm{\lambda}_\mathrm{out}}{\text{maximize}} & \sum_{i \in \mathcal{V}} \lambda^\mathrm{in}_i \sigma^\mathrm{in}_i + \sum_{j \in \mathcal{V}} \lambda^\mathrm{out}_j \sigma^\mathrm{out}_j \\
\text{subject to} &  \lambda^\mathrm{in}_i + \lambda^\mathrm{out}_j \leq c_{ij} \qquad \forall i,j \in \mathcal{V}.\\
\end{array} \label{dual_alt_prob}
\end{equation}
Here, the dual vectors $\bm{\lambda}_\mathrm{in} = (\lambda^\mathrm{in}_i)$ and $\bm{\lambda}_\mathrm{out} = (\lambda^\mathrm{out}_i)$ can be interpreted respectively as the dual \emph{embarkment prices} on \emph{sources} and \emph{disembarkment prices} on \emph{targets}, as shown in \cite{villani2003topics,villani2008optimal}. This is a common property of the dual problem in linear programming \cite{griva2009linear}.

\subsection{The standard optimal transport flow on a graph problem}
\label{SubSec_optimal_transport_graph_problem01}

For completeness, let us recall the standard (exact) transport flow problem on a graph.
The linear programming optimal transport flow problem is defined as \cite{ahuja1993network}
\begin{equation}
\vline \begin{array}{ll@{}ll}
\underset{\mathbf{X}}{\text{minimize}} & \mathbf{e}^\top ( \mathbf{X} \circ \mathbf{C}' ) \mathbf{e} \\
\text{subject to}& \left( \mathbf{X} -  \mathbf{X}^\top \right) \mathbf{e} = \bm{\sigma}_\mathrm{in} -  \bm{\sigma}_\mathrm{out}, \\
& x_{ij} \geq 0 \qquad \quad \forall i,j \in \mathcal{V},
\end{array}
\label{Eq_exact_optimal_transport_LP01}
\end{equation}
where $\mathbf{C}'$ is the cost matrix containing non-negative costs on the edges of the network with infinite components set to $0$. As before, it is assumed that input and output flows are non-negative as well as $(\bm{\sigma}_\mathrm{in})^\top \mathbf{e} = (\bm{\sigma}_\mathrm{out})^\top \mathbf{e} = 1$. The idea is therefore to minimize the total cost of flows while satisfying the input and output constraints. The solution $\mathbf{X} = (x_{ij})$ of this problem corresponds to the matrix containing directed flows on the edges.

\subsection{The regularized optimal transport problem}
\label{Subsec_regularized_optimal_transport01}

In order to show that the bag-of-paths formalism is closely related to the optimal transport problem, we need to compute again the minimum free energy value in terms of the elements of the  coupling matrix $\gamma_{ij}$, instead of the Lagrange parameters $\lambda_{i}^{\mathrm{in}}, \lambda_{j}^{\mathrm{out}}$ as in (\ref{FE}). 

Let us start with the regular, non-hitting paths model.
By inserting the form taken by the optimal path probability distribution $\mathrm{P}^\star(\wp)$, given by (\ref{prob_ref_comp}), in the free energy functional (\ref{init_prob}) (similarly to the proof following Proposition \ref{prop::optimal_FE}), we obtain 
\begin{gather*}
\text{FE}(\mathrm{P}^\star) = T \sum_{i,j \in \mathcal{V}} \log( \mu^\mathrm{in}_i \mu^\mathrm{out}_j) \underbracket[0.5pt][4pt]{ \sum_{\wp_{ij} \in \mathcal{P}_{ij}} \mathrm{P}^\star(\wp_{ij}) }_{\gamma_{ij}},
\end{gather*}
which directly provides from $\mu^\mathrm{in}_i \mu^\mathrm{out}_j = \gamma_{ij} / (\sigma^\mathrm{in}_{i}  \alpha_j \widehat{z}_{ij})$ (see Equation (\ref{pi_element})) and $\gamma_{ij} = \sum_{\wp_{ij} \in \mathcal{P}_{ij}} \mathrm{P}^\star(\wp_{ij})$ (Equation (\ref{Eq_coupling_probabilities01})) 
\begin{equation}
\text{FE}(\mathrm{P}^\star)
= -T \sum_{i,j \in \mathcal{V}} \log( \widehat{z}_{ij} ) \, \gamma_{ij} + T \sum_{i,j \in \mathcal{V}} \gamma_{ij} \log\left( \frac{\gamma_{ij}}{\sigma^\mathrm{in}_i \alpha_j } \right).
\label{altern_FE}
\end{equation}
Following the same reasoning for the hitting paths formalism, we obtain
\begin{equation}
\text{FE}_\mathrm{h}(\mathrm{P}^\star_\mathrm{h}) = \sum_{i,j \in \mathcal{V}} \underbracket[0.5pt][4pt]{ -T \log( z^\mathrm{h}_{ij} ) }_{\phi_T(i,j)} \, \gamma^\mathrm{h}_{ij} + T \sum_{i,j \in \mathcal{V}} \gamma^\mathrm{h}_{ij} \log\left( \frac{\gamma^\mathrm{h}_{ij}}{\sigma^\mathrm{in}_i \sigma^\mathrm{out}_j } \right). \label{altern_FE_h}
\end{equation}
This shows that, instead of working with the whole probability distribution as required by Equation (\ref{init_prob}), it is sufficient to compute the elements of the coupling matrix $\gamma_{ij}$ if the quantities $\phi_T(i,j)$ are pre-computed.

Indeed, the resulting expression in the hitting formalism has a nice interpretation. Indeed, in \cite{kivimaki2014developments} it was shown that in the simple randomized shortest-paths framework $\phi_T(i,j) \triangleq \mathrm{FE}_\mathrm{h}(\mathrm{P}^\star_\mathrm{h}) = -T \log( z^\mathrm{h}_{ij} )$ is known to be the minimum free energy when problem (\ref{old_prob}) is restricted to hitting paths connecting a single source $i$ to a single destination $j$. This quantity corresponds to the pairwise directed \emph{free energy distance} between nodes of a graph introduced in \cite{kivimaki2014developments,franccoisse2016bag}, where it is proved that it is a distance metric. This distance provided competitive results in pattern recognition tasks \cite{franccoisse2016bag,Sommer-2016,Sommer-2017}.

Moreover, it is shown in \cite{franccoisse2016bag} that $\phi_T(i,j)$ converges to the directed \emph{shortest path distance} between $i$ and $j$ when $T \to 0$, and to the \emph{average first passage time} (up to a scaling factor) between $i$ and $j$ when $T \to \infty$. The free energy distances between all pairs of nodes can easily be computed in matrix form \cite{kivimaki2014developments,franccoisse2016bag,fouss2016algorithms}. 
It has further been shown that, when computing the continuous time -- continuous state equivalent to the randomized shortest-paths model by densifying the graph, the minimum free energy becomes a \emph{potential} attracting the agents to the goal state \cite{Garcia-Diez-2011b}. 

From (\ref{altern_FE_h}), once the directed free energy distances have been computed, we observe that problem (\ref{init_prob_h}) can be restated as
\begin{equation}
\vline \begin{array}{ll@{}ll}
\underset{\bm{\Gamma_\mathrm{h}}}{\text{minimize}} & \sum_{i,j \in \mathcal{V}} \phi_T(i,j) \gamma^\mathrm{h}_{ij} + T \sum_{i,j \in \mathcal{V}} \gamma^\mathrm{h}_{ij} \log\left( \dfrac{\gamma^\mathrm{h}_{ij}}{\sigma^\mathrm{in}_i \sigma^\mathrm{out}_j } \right) \\
\text{subject to} & \sum_{j \in \mathcal{V}} \gamma^\mathrm{h}_{ij}= \sigma^\mathrm{in}_i \qquad \; \forall i \in \mathcal{V},\\
          & \sum_{i \in \mathcal{V}} \gamma^\mathrm{h}_{ij} = \sigma^\mathrm{out}_j \qquad \forall j \in \mathcal{V}, \\
          & \gamma^\mathrm{h}_{ij} \geq 0 \qquad \qquad \qquad \forall i,j \in \mathcal{V}.
\end{array} \label{altern_prob_hitting}
\end{equation}
Knowing that the directed free energy distance $\phi_T(i,j)$ converges to the directed shortest path distance when $T \to 0$, we conclude that problem (\ref{altern_prob_hitting}) reduces to the optimal transport problem at this limit. Thus, problem (\ref{altern_prob_hitting}) is actually a ``soft" (entropy regularized) optimal transport problem, similar to the one studied in, e.g., \cite{Wilson-1970,Erlander-1990,Cuturi2013} and based on the directed free energy distance \cite{kivimaki2014developments,franccoisse2016bag,fouss2016algorithms}, monitored by temperature.

Therefore, in the case of hitting paths, an alternative way of solving the entropy-regularized optimal transport on a graph problem (\ref{init_prob_h}) is to \emph{pre-compute the free energy distances} $\phi_T(i,j)$ and then solve problem (\ref{altern_prob_hitting}) (see \cite{Cuturi2013} for a recent discussion). 

The non-hitting problem is another regularization of the optimal transport problem as both (\ref{altern_FE}) and (\ref{altern_FE_h}) converge to the same limit when $T \to 0$. As a matter of fact, when $T \to 0$, $\widehat{z}_{ij}$ and $z^\mathrm{h}_{ij}$ have the same limit, as the effect of $\alpha_i$ vanishes in the first case and $z_{jj} \to 1$ in the second case. In contrast, these two formalisms diverge when $T \to \infty$. The hitting paths formalism converges to the problem described in \cite{guex2017transportation}, and thus results in the trivial, independent, coupling $\gamma^\mathrm{h}_{ij} = \sigma^\mathrm{in}_i  \sigma^\mathrm{out}_j$ when $T \to \infty$, while the non-hitting paths formalism provides a more interesting solution, though harder to interpret (see Equation (\ref{altern_FE})).
Actually, following the derivations appearing so far in this paper and results discussed in \cite{guex2016interpolating}, it appears that the high temperature limit of the non-hitting formalism corresponds, in the case of an undirected graph, to the \emph{electrical circuit formalism}, where sources and targets correspond to nodes with potentials fixed by the user (a high potential on sources and a low potential on targets) (see \cite{guex2016interpolating} for details). A study of this interesting question is left for further work.

Interestingly, Equations (\ref{FE}) and (\ref{FE_h}) give two different, alternative, expressions for the minimal free energy, which are equivalent to the objective function of the dual optimal transport problem (with Lagrange multipliers corresponding to the dual variables multiplied by $-1$). It implies that, when $T \to 0$, $-\bm{\lambda}_\mathrm{in}$ and $-\bm{\lambda}_\mathrm{out}$ (and equivalently, $-\bm{\lambda}_\mathrm{in}^\mathrm{h}$ and $-\bm{\lambda}_\mathrm{out}^\mathrm{h}$) converge respectively to the dual embarkment prices on sources and disembarkment prices on targets. Therefore, a mapping of these variables on nodes can highlight problematic sources and targets, in terms of optimal transport (see, e.g., \cite{guex2017transportation}).

\subsection{The optimal transportation flow}

Within the context of the optimal transport problem, it is interesting to discuss the interpretation of the matrix containing the expected number of visits to edges, i.e., $\mathbf{\bar{N}} = (\bar{n}_{ij})$. Note that the discussion is developed within the non-hitting formalism, but remains valid in the hitting case.

In Section \ref{sec::nnodes}, the number of times node $i$ appears on a path $\wp$ was decomposed in $\eta(i \in \wp) = \sum_{j \in \mathcal{V}} \eta(i \to j \in \wp) + \delta(E(\wp) = i)$ (Equation (\ref{eta_dev})), leading to $\bar{n}_i =  \sum_{j \in \mathcal{V}} \bar{n}_{ij} + \sigma^\mathrm{out}_i$ (Equation (\ref{n_i_sum})). However, it is also possible to write 
\begin{equation*}
\eta(i \in \wp) = \sum_{j \in \mathcal{V}} \eta(j \to i \in \wp) + \delta(S(\wp) = i),
\end{equation*}
where $\delta(S(\wp) = i)$ is 1 iff $\wp$ begins with $i$ and 0 otherwise. This second version results in having
\begin{equation*}
\bar{n}_i = \sum_{j \in \mathcal{V}} \bar{n}_{ji} + \sigma^\mathrm{in}_i,
\end{equation*}
and combining this result with (\ref{n_i_sum}), we obtain
\begin{equation*}
\sum_{j \in \mathcal{V}} \bar{n}_{ij} - \sum_{j \in \mathcal{V}} \bar{n}_{ji} =  \sigma^\mathrm{in}_i -  \sigma^\mathrm{out}_i, \; \forall i \in \mathcal{V}.
\end{equation*}
This last equation shows that $\bar{n}_{ij}$ can be interpreted as a \emph{directed flow on edges}, as shown in \cite{guex2016interpolating}. This flow is emitted by sources, absorbed by targets, and conserved everywhere else. Mapping this flow allows us to analyse the transportation of matter along the edges of the graph, and is illustrated in Section \ref{Sec_case_studies01}. In \cite{guex2016interpolating}, it is shown that the \emph{net flow} on an undirected graph, i.e.\ $\bar{n}_{ij} - \bar{n}_{ji}$, converges to the \emph{electrical flow} for the non-hitting formalism when $T \to \infty$. Furthermore, in the hitting formalism, summing the absolute values of the net flows over the edges results in a \emph{weighted randomized shortest path} (RSP) \emph{net betweenness centrality} \cite{kivimaki2016two}. For an unweighted graph, the standard, unweighted RSP net betweenness converges to the \emph{current flow betweenness} in the limit $T \to \infty$ \cite{newman2005measure, Brandes-2005b}.


\section{Derived distances and dissimilarities}
\label{Sec_distance01}

Two general families of distances are derived from our framework: distances between nodes and dissimilarities between groups of nodes (both for the non-hitting and the hitting case).

\subsection{Distances between nodes}

Let us first discuss distances between nodes.

\subsubsection{Definitions}

For both the hitting and non-hitting formalism, we can now define a distance named the \emph{surprisal distance}, generalizing the one introduced in \cite{franccoisse2016bag,kivimaki2014developments}. The particularity here is that we can attach positive \emph{weights} $\mathbf{w} = (w_i)$ (with $w_i > 0$ and summing to 1) to nodes, which affect the distances through $\bm{\sigma}_\mathrm{in}$ and $\bm{\sigma}_\mathrm{out}$. More precisely, for a strongly connected graph, we define the \emph{margin-constrained bag-of-paths surprisal distance} $\bm{\Delta}_\mathrm{sur} = (\dist^\mathrm{sur}_{ij})$, and \emph{margin-constrained bag-of-hitting-paths surprisal distance}, $\bm{\Delta}^\mathrm{h}_\mathrm{sur} = (\dist^\mathrm{h,sur}_{ij})$ by, respectively,
\begin{align}
&\dist^\mathrm{sur}_{ij} \triangleq \left\{
	\begin{array}{ll}
		-\frac{1}{2}(\log(\gamma_{ij}) + \log(\gamma_{ji}))  & \mbox{if } i\neq j, \\
		\phantom{-} 0 & \mbox{if } i = j,
	\end{array} \right. \label{surdist} \\
	~\notag \\
&\dist^\mathrm{h,sur}_{ij} \triangleq \left\{
	\begin{array}{ll}
		-\frac{1}{2}(\log(\gamma^\mathrm{h}_{ij}) + \log(\gamma^\mathrm{h}_{ji}))  & \mbox{if } i\neq j, \\
		\phantom{-} 0 & \mbox{if } i = j,
	\end{array} \right. \label{surdist_h}
\end{align}
where $\mathbf{\Gamma} = (\gamma_{ij})$ and $\mathbf{\Gamma}_\mathrm{h} = (\gamma_{ij}^\mathrm{h})$ (the elements of the coupling matrix, see Sections \ref{Subsec_coupling_matrix_regular_paths01} and \ref{Subsec_coupling_matrix_hitting_paths01}) are obtained from, respectively, the non-hitting and hitting path formalisms with $\bm{\sigma}_\mathrm{in} = \bm{\sigma}_\mathrm{out} = \mathbf{w}$. From this definition, each node acts as a source and a target and larger weights induce a stronger influence over the graph, as the flows along the paths starting and ending in a particular node will scale accordingly. 

\begin{proposition}
The margin-constrained bag-of-paths surprisal distance, $\bm{\Delta}_\mathrm{sur}$, defined in (\ref{surdist}), and the margin-constrained bag-of-hitting-paths surprisal distance, $\bm{\Delta}^\mathrm{h}_\mathrm{sur}$ defined in (\ref{surdist_h}), are both metrics.
\end{proposition}
\begin{proof}
The triangle inequality for both surprisal distances, i.e., $\dist^\mathrm{sur}_{ij} \leq \dist^\mathrm{sur}_{ik} + \dist^\mathrm{sur}_{kj}$ is trivially proven if $i=k$, $i = k$ or $j = k$. So we will assume here that $i \neq j \neq k \neq i$.

\paragraph{Non-hitting formalism.}

First, notice that the reasoning found in Appendix B of \cite{franccoisse2016bag} is still valid with the non-hitting reference probabilities derived in Section \ref{Sec_non_hitting_reference_probabilities01}, namely, we have $\widehat{\pi}^\mathrm{ref}(\wp_{ij}) = \widehat{\pi}^\mathrm{ref}(\wp^{\mathrm{h}}_{ij}) \widehat{\pi}^\mathrm{ref}(\wp_{jj})$ and $\tilde{c}(\wp_{ij}) = \tilde{c}(\wp^{\mathrm{h}}_{ij}) + \tilde{c}(\wp_{jj})$, where $\wp^{\mathrm{h}}_{ij}$ is the hitting path consisting of the first part of $\wp_{ij}$, until it reaches $j$ for the first time, and $\wp_{jj}$ is the remaining part of $\wp_{ij}$. Thus, we also have
\begin{gather}
\widehat{z}_{ij} = \widehat{z}^\mathrm{h}_{ij} \widehat{z}_{jj}, \label{z_trick}
\end{gather}
where $\widehat{z}^\mathrm{h}_{ij} \triangleq \sum_{\wp^\mathrm{h}_{ij} \in \mathcal{P}^\mathrm{h}_{ij}} \widehat{\pi}^\mathrm{ref}(\wp^\mathrm{h}_{ij}) \exp( - \beta \tilde{c}(\wp^\mathrm{h}_{ij}) )$. Now, it is easy to see that, for the optimal path probabilities obtained in (\ref{prob_ref_comp}),
\begin{equation}
\mathrm{P}(S = i, E = j) \geq \mathrm{P}(S = i, k \in \wp, E = j) \triangleq \sum_{\wp_{ij} \in \mathcal{P}_{ij}} \delta(k \in \wp_{ij}) \mathrm{P}^\star(\wp_{ij}), \label{int_def}
\end{equation}
where $\delta(k \in \wp_{ij})$ is equal to $1$ if node $k$ lies on path $\wp_{ij}$ and $0$ otherwise.
By developing with (\ref{prob_ref_comp}), we obtain
\begin{align*}
&\mathrm{P}(S = i, k \in \wp, E = j)
= \mu^\mathrm{in}_{i}  \sigma^\mathrm{in}_{i} \mu^\mathrm{out}_{j}   \alpha_{j} \sum_{\wp_{ij} \in \mathcal{P}_{ij}} \delta(k \in \wp_{ij}) \widehat{\pi}^\mathrm{ref}(\wp_{ij}) \exp(-\beta \tilde{c}(\wp_{ij} )) \\
&= \mu^\mathrm{in}_{i}  \sigma^\mathrm{in}_{i} \mu^\mathrm{out}_{j}   \alpha_{j} \sum_{\wp^\mathrm{h}_{ik} \in \mathcal{P}^\mathrm{h}_{ik}} \tilde{\pi}^\mathrm{ref}(\wp^\mathrm{h}_{ik}) \exp(-\beta \tilde{c}(\wp^\mathrm{h}_{ik} )) \sum_{\wp_{kj} \in \mathcal{P}_{kj}} \widehat{\pi}^\mathrm{ref}(\wp_{kj}) \exp(-\beta \tilde{c}(\wp_{kj} )) \\
&= \mu^\mathrm{in}_{i}  \sigma^\mathrm{in}_{i} \mu^\mathrm{out}_{j}  \alpha_{j} \widehat{z}^\mathrm{h}_{ik} \widehat{z}_{kj},
\end{align*}
where each path from $i$ to $j$ is again cut in two sequential sub-paths.

Then, using Equations (\ref{z_trick}) and (\ref{pi_element}) provides
\begin{align*}
\mathrm{P}(S = i, k \in \wp, E = j) &= \frac{\mu^\mathrm{in}_{i} \sigma^\mathrm{in}_{i}  \mu^\mathrm{out}_{k}  \alpha_{k} \widehat{z}_{ik} \mu^\mathrm{in}_{k} \sigma^\mathrm{in}_{k} \mu^\mathrm{out}_{j}  \alpha_{j} \widehat{z}_{kj}}{ \mu^\mathrm{in}_{k} \sigma^\mathrm{in}_{k}  \mu^\mathrm{out}_{k}  \alpha_{k}  \widehat{z}_{kk}} \\
&= \frac{ \mathrm{P}(S = i, E = k) \mathrm{P}(S = k, E = j) }{\mathrm{P}(S = k, E = k) }.
\end{align*}
Using the inequality in (\ref{int_def}) and $\mathrm{P}(S = k, E = k) \leq 1$, we get 
\begin{gather*}
-\log(\mathrm{P}(S = i, E = j)) \leq -\log(\mathrm{P}(S = i, E = k)) -\log(\mathrm{P}(S = k, E = j)),
\end{gather*}
which proves the triangle inequality for the non-hitting surprisal distance for $i \neq j \neq k \neq i$.

\paragraph{Hitting formalism.}

The reasoning is similar to the previous case. First let us consider 
\begin{align}
\mathrm{P}_\mathrm{h}(S = i, E = j) &=  \mu^\mathrm{h,in}_{i}  \sigma^\mathrm{in}_{i} \mu^\mathrm{h,out}_{j} \sigma^\mathrm{out}_{j} z^\mathrm{h}_{ij} = \mu^\mathrm{h,in}_{i}  \sigma^\mathrm{in}_{i} \mu^\mathrm{h,out}_{j} \sigma^\mathrm{out}_{j} \frac{z_{ij}}{z_{jj}}  \notag \\ 
&\geq  \mu^\mathrm{h,in}_{i}  \sigma^\mathrm{in}_{i} \mu^\mathrm{h,out}_{j} \sigma^\mathrm{out}_{j} \frac{\sum_{\wp_{ij} \in \mathcal{P}_{ij}} \delta(k \in \wp_{ij}) \tilde{\pi}^\mathrm{ref}(\wp_{ij}) \exp(-\beta \tilde{c}(\wp_{ij} )) }{z_{jj}}. \label{int_def2}
\end{align}
And again, 
\begin{align*}
&\mu^\mathrm{h,in}_{i}  \sigma^\mathrm{in}_{i} \mu^\mathrm{h,out}_{j} \sigma^\mathrm{out}_{j} \frac{\sum_{\wp_{ij} \in \mathcal{P}_{ij}} \delta(k \in \wp_{ij}) \tilde{\pi}^\mathrm{ref}(\wp_{ij}) \exp(-\beta \tilde{c}(\wp_{ij} )) }{z_{jj}} \\
&=\frac{\mu^\mathrm{h,in}_{i}  \sigma^\mathrm{in}_{i} \mu^\mathrm{h,out}_{j} \sigma^\mathrm{out}_{j}}{z_{jj}} \sum_{\wp^\mathrm{h}_{ik} \in \mathcal{P}^\mathrm{h}_{ik}} \tilde{\pi}^\mathrm{ref}(\wp^\mathrm{h}_{ik}) \exp(-\beta \tilde{c}(\wp^\mathrm{h}_{ik} ))  \sum_{\wp_{kj} \in \mathcal{P}_{kj}}  \tilde{\pi}^\mathrm{ref}(\wp_{kj}) \exp(-\beta \tilde{c}(\wp_{kj} )) \\
&= \frac{ \mu^\mathrm{h,in}_{i}  \sigma^\mathrm{in}_{i} \mu^\mathrm{h,out}_{k} \sigma^\mathrm{out}_{k} z^\mathrm{h}_{ik} \; \mu^\mathrm{h,in}_{k}  \sigma^\mathrm{in}_{k} \mu^\mathrm{h,out}_{j} \sigma^\mathrm{out}_{j} z^\mathrm{h}_{kj} }{ \mu^\mathrm{h,in}_{k}  \sigma^\mathrm{in}_{k} \mu^\mathrm{h,out}_{k} \sigma^\mathrm{out}_{k} }=  \frac{ \mathrm{P}_\mathrm{h}(S = i, E = k) \mathrm{P}_\mathrm{h}(S = k, E = j) }{\mathrm{P}_\mathrm{h}(S = k, E = k) },
\end{align*}
where we used $z^\mathrm{h}_{kk} = 1$. With (\ref{int_def2}) and $\mathrm{P}_\mathrm{h}(S = k, E = k) \leq 1$, we finally obtain
\begin{gather*}
-\log(\mathrm{P}_\mathrm{h}(S = i, E = j)) \leq -\log(\mathrm{P}_\mathrm{h}(S = i, E = k)) -\log(\mathrm{P}_\mathrm{h}(S = k, E = j)),
\end{gather*}
which shows the triangle inequality for the hitting surprisal distance for $i \neq j \neq k \neq i$.
\end{proof}

\subsection{Distances between groups of nodes}
\label{sec::dissgroups}

A different family of dissimilarities naturally arises from the optimal transport interpretation of the margin-constrained BoP formalism, namely dissimilarities between \emph{groups of nodes}. These dissimilarities can be viewed as an extension of the \emph{Wasserstein distance} between node distributions on a graph, also known under the name of the \emph{Monge-Kantorovich distance} or the \emph{earth mover distance} in the literature (see e.g. \cite{dobrushin1970prescribing,villani2003topics,villani2008optimal,zolotarev1983probability}).

This dissimilarity is defined as the total cost of transportation in order to move from the distribution on source nodes, $\bm{\sigma}_\mathrm{in}$, to the distribution on target nodes, $\bm{\sigma}_\mathrm{out}$. Similarly to the usual free energy distance described in \cite{kivimaki2014developments,franccoisse2016bag},  interpolating between the shortest path distance and the commute cost distance (which is proportional to the resistance distance for undirected graphs), the margin-constrained BoP formalism uses the value of the free energy functional in order to derive a dissimilarity which interpolates between the Wasserstein distance and an electrical circuit-based dissimilarity between groups of node.

\subsubsection{Definitions} 

Let $\mathcal{G}$ be a directed, strongly connected, graph with $n$ nodes, weighted by vector $\mathbf{w}$. Suppose we have $p$ groups of nodes, and the \emph{membership matrix} $\mathbf{M} = (m_{ig})$ with $m_{ig} \geq 0$ and $\sum_{g=1}^p m_{ig} = 1$, represents the membership degree of node $i$ to group $g$ (fuzzy memberships are allowed). From that, we can compute the \emph{node distribution in group g}, $\bm{\sigma}_g = ({\sigma}^g_i)$, as , e.g.,
\begin{equation*}
{\sigma}^g_i = \frac{w_i m_{ig}}{\sum_{j \in \mathcal{V}} w_j m_{jg}}.
\end{equation*}
Then, as for the standard free energy distance between two nodes \cite[see this paper for details]{kivimaki2014developments}, the \emph{bag-of-paths free energy dissimilarity} between groups $g$ and $h$ is defined as the symmetrized minimum free energy between these two groups of nodes 
\begin{equation}
\dist^\text{FE}_{gh} \triangleq \left\{ \begin{array}{ll}
		\frac{1}{2}( \text{FE}^\mathrm{min}(\bm{\sigma}_g,\bm{\sigma}_h) +  \text{FE}^\mathrm{min}(\bm{\sigma}_h,\bm{\sigma}_g) )  & \mbox{if } i\neq j, \\
		0 & \mbox{if } i = j,
	\end{array} \right. 
\end{equation}
and the \emph{bag-of-hitting-paths free energy dissimilarity} between groups $g$ and $h$ by 
\begin{equation}
\dist^\mathrm{h,FE}_{gh} \triangleq \left\{ \begin{array}{ll}
		\frac{1}{2}( \text{FE}_\mathrm{h}^\mathrm{min}(\bm{\sigma}_g,\bm{\sigma}_h) +  \text{FE}_\mathrm{h}^\mathrm{min}(\bm{\sigma}_h,\bm{\sigma}_g) )  & \mbox{if } i\neq j, \\
		0 & \mbox{if } i = j,
	\end{array} \right. 
\end{equation}
where $\text{FE}^\mathrm{min}(\bm{\sigma}_\mathrm{in},\bm{\sigma}_\mathrm{out})$ is the non-hitting free energy (\ref{FE}), and $\text{FE}_\mathrm{h}^\mathrm{min}(\bm{\sigma}_\mathrm{in},\bm{\sigma}_\mathrm{out})$ the hitting free energy (\ref{FE_h}), with starting and ending node flows respectively equal to $\bm{\sigma}_\mathrm{in}$ and $\bm{\sigma}_\mathrm{out}$. By definition of the free energy, we are sure that this quantity is always positive. When $T \to 0$, the dissimilarity between groups $g$ and $h$ will yield the \emph{optimal cost of transportation} from group $g$ to $h$ and from group $h$ to $g$, which is obviously a metric \cite{dobrushin1970prescribing,villani2003topics,villani2008optimal,zolotarev1983probability}. It is, however, possible that this dissimilarity is not a metric anymore for other values of $T$.

From a computational point of view, there exists an important difference between the bag-of-hitting-paths and the bag-of-paths algorithms computing their respective free energy dissimilarities. As a matter of fact, in the hitting formalism, the matrix $\mathbf{Z}_\mathrm{h}$ is only computed once, and dissimilarities between each pair of groups can be obtained afterward by solely changing the values of $\bm{\sigma}_\mathrm{in}$ and $\bm{\sigma}_\mathrm{out}$ in the iterative procedure defined by (\ref{it1_h}) and (\ref{it2_h}). It is however impossible to proceed that way for the non-hitting formalism, as the computation of $\widehat{\mathbf{Z}}$ requires the values of $\bm{\sigma}_\mathrm{in}$ and $\bm{\sigma}_\mathrm{out}$. Therefore, the bag-of-hitting-paths distances are obtained in a more efficient way than the bag-of-paths distances.


\begin{figure}
\centering
\begin{subfigure}[]{0.45\textwidth}
      \includegraphics[width = \textwidth, trim=3cm 3cm 2.3cm 2.5cm, clip=true]{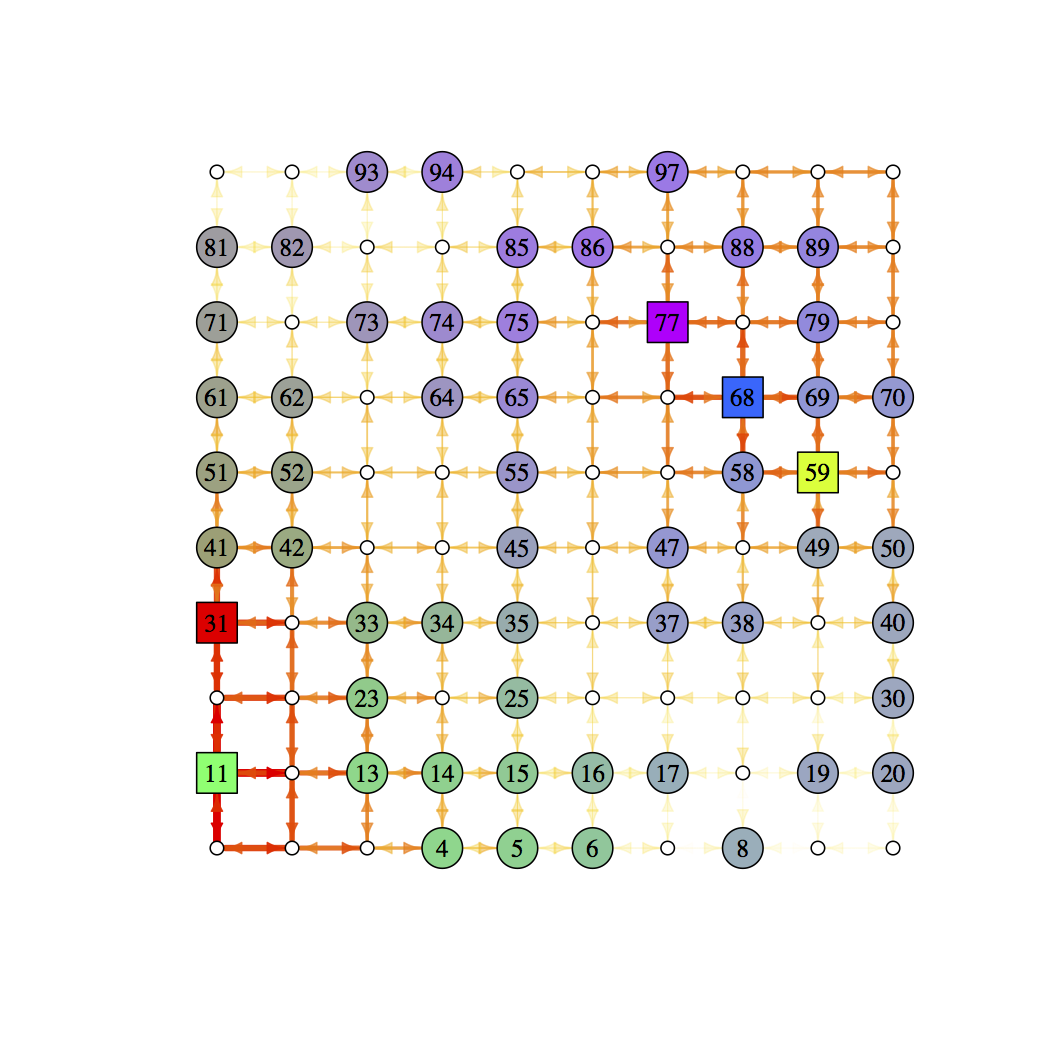}
      \caption{cBoP, $\beta=10^{-3}$}
\end{subfigure}
\qquad
\begin{subfigure}[]{0.45\textwidth}
      \includegraphics[width = \textwidth, trim=3cm 3cm 2.3cm 2.5cm, clip=true]{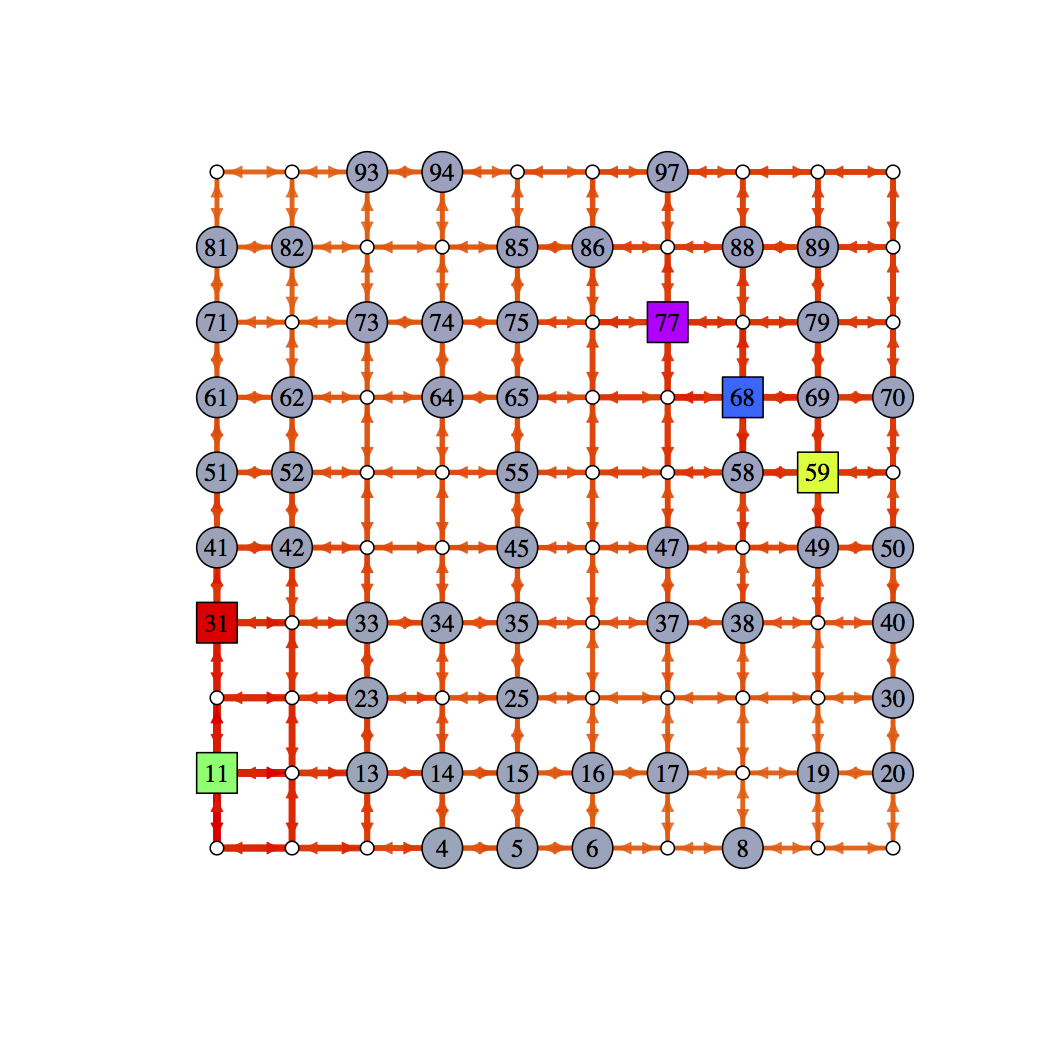}
      \caption{cBoHP, $\beta=10^{-3}$}
\end{subfigure}
\begin{subfigure}[]{0.45\textwidth}
      \includegraphics[width = \textwidth, trim=3cm 3cm 2.3cm 2.5cm, clip=true]{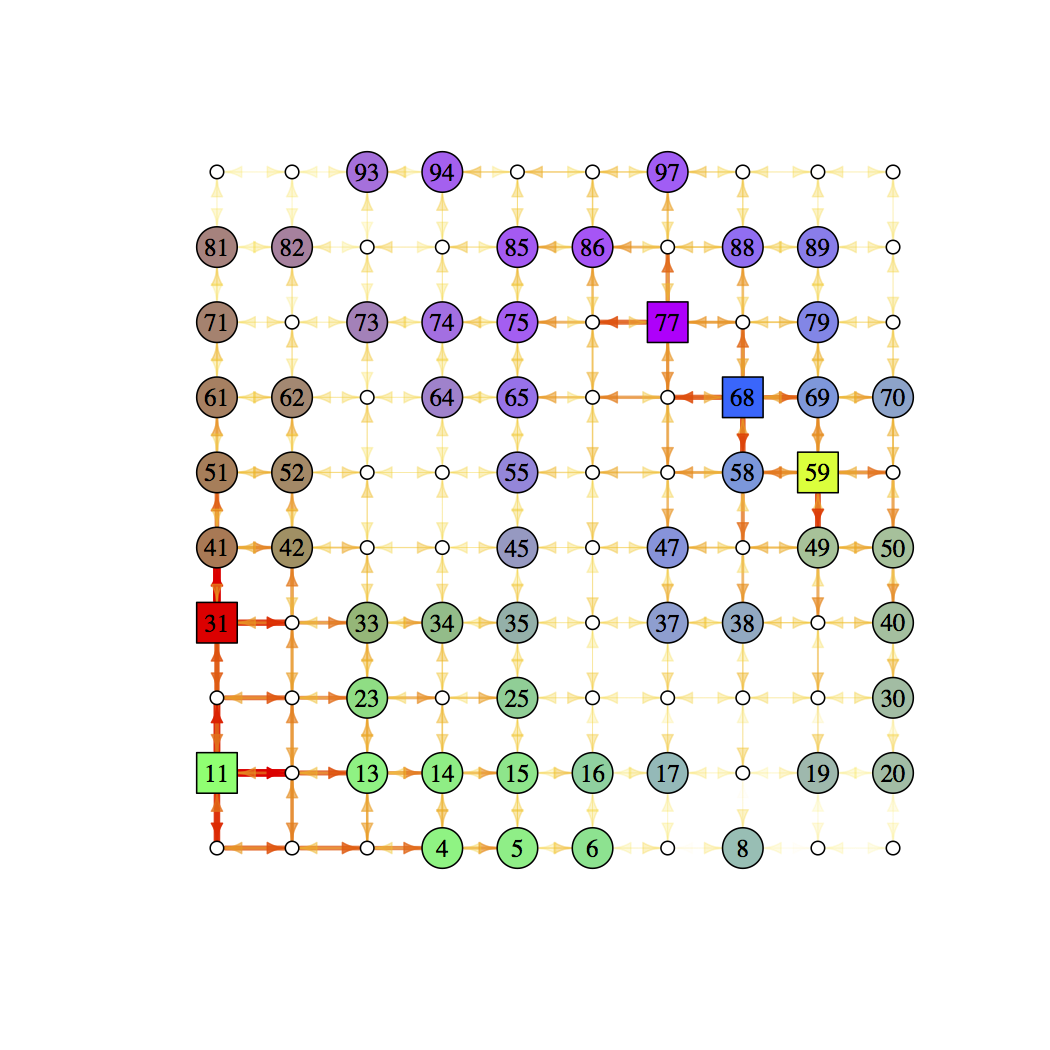} 
      \caption{cBoP, $\beta=10^{-1}$}
\end{subfigure}
\qquad
\begin{subfigure}[]{0.45\textwidth}
      \includegraphics[width = \textwidth, trim=3cm 3cm 2.3cm 2.5cm, clip=true]{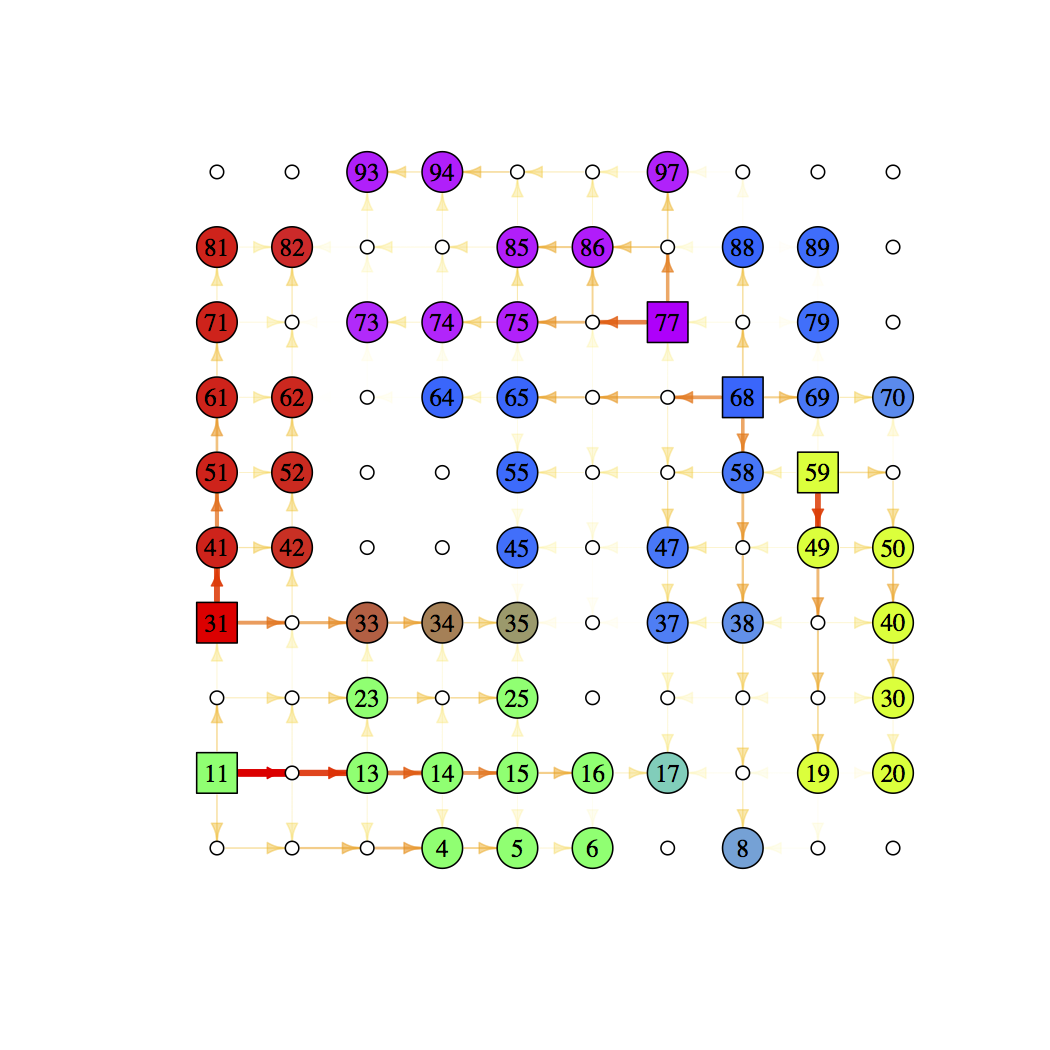}
      \caption{cBoP, $\beta=10$}
\end{subfigure}
\caption{Resulting coupling $\bm{\Gamma}$ and edge flows $\mathbf{\bar{N}}$ for the margin-constrained bag-of-paths (cBoP) and the margin-constrained bag-of-hitting-paths (cBoHP) models, for different values of $\beta$. Source nodes are represented by a square and targets by a circle. $\bm{\Gamma}$ is shown on targets by a color interpolation representing their membership to corresponding sources (the membership of target $j$ to source $i$ is $\mathrm{P}(S=i | E=j) = \gamma_{ij} / \sigma^\mathrm{out}_j$). $\mathbf{\bar{N}}$ is represented on edges.}
\label{fig:toy_res}
\end{figure}

\section{Some illustrations}
\label{Sec_case_studies01}
\label{sec::opttranspexp}

Although the main contribution of this work lies in the theoretical development of the margins constrained bag-of-paths models, we provide here an illustration of the algorithms on a toy example. The bag-of-paths formalism (as well as the bag-of-hitting-paths, as they converge to the same solution when $T\to 0$) defines an efficient way to find an approximate solution of the transportation problem on a graph. Moreover, by varying the temperature of the model, we can add uncertainty to the optimal paths and offer a flexible, stochastic, alternative to the optimal solution. With this feature, we can, e.g., pinpoint target nodes for which the optimal source coupling is most unclear. In a practical setting, when focusing on a deterministic transportation policy, this can help in evaluating the importance of each source-target coupling decision. Unlike other efficient optimal transport solvers, the algorithm here provides not only the coupling between pairs of nodes, but also the flows on edges. Knowing the most frequented edges could be a major asset for real-life applications, for example in order to forecast network traffic. This section illustrates this idea on a toy graph (a $10 \times 10$ lattice).
In addition, we evaluate the computational efficiency of Algorithms \ref{Alg_nh_bop} and \ref{Alg_h_bop} by comparing the computation times with a baseline linear solver on lattices of different sizes.

\subsection{Illustrations on a $10 \times 10$ lattice}

In this illustrative example, 5 random nodes were picked with $\sigma^\mathrm{in}_i = 0.2$ and 50 others with $\sigma^\mathrm{out}_i = 0.02$. The resulting coupling $\mathbf{\Gamma}$ and flows on edges $\mathbf{\bar{N}}$ are represented in Figure \ref{fig:toy_res} for the constrained bag-of-paths (cBoP) and the constrained bag-of-hitting-paths (cBoHP), with different values of $\beta$. In this figure, target nodes $j$ are colored to represent their distribution of membership over sources, i.e.\ $\mathrm{P}(S=i | T=j) = \gamma_{ij} / \sigma^\mathrm{out}_j$, and edge colors display the flows, $\mathbf{\bar{N}}$.

We observe that results obtained by the bag-of-paths model and the bag-of-hitting-paths model are quite different for $\beta=10^{-3}$ (Figure \ref{fig:toy_res}, top row): the bag-of-paths model displays a behavior similar to a diffusive process, with edges near sources drawn more frequently, which is known to be similar to the electrical current \cite{ahuja1993network, Snell-1984, guex2016interpolating}. On the other hand, the bag-of-hitting-paths solution for $\beta=10^{-3}$ is quite trivial, with a uniform distribution of memberships of every target $j$ to sources and the flow almost similar on every edge. In contrast, when the temperature is low, both models converge to the same solution and, to avoid redundancy, only the bag-of-paths model is shown here (Figure \ref{fig:toy_res}, bottom row). With $\beta=10$, this model displays an optimal transport solution, with only shortest paths followed and almost deterministic distributions of targets-to-sources memberships.

Therefore, in the present problem, the constrained bag-of-hitting-paths is perhaps less useful than the constrained bag-of-paths when the parameter $\beta$ is close to zero. However, this depends on the application at hand and, essentially, on the desired behaviour of the system when $\beta \rightarrow 0$, either the solution of an electrical circuit or the independence between sources and destinations.

\subsection{Comparison of computation time}
\label{sec::speed}
To assess the computation time of the constrained bag-of-paths and bag-of-hitting-paths algorithms, we compare them to the open-source \emph{Computational Infrastructure for Operations Research} (Coin-or) linear programming solver (\texttt{clp}) written in C++ \cite{heimer2003common}, which is considered as an efficient baseline algorithm for finding the coupling and the flow of the exact optimal transport problem (see Equation (\ref{Eq_exact_optimal_transport_LP01})). We run the algorithms on lattices of various dimensions in order to increase the number of nodes $n$. The number of source nodes and target nodes are set to be both $\lfloor n/3 \rfloor$ (rounded down) with similar weights, and their locations are randomly selected.

\begin{figure}[t]
\begin{center}
      \includegraphics[width = 0.9\textwidth, trim=0cm 0cm 0cm 0cm, clip=true]{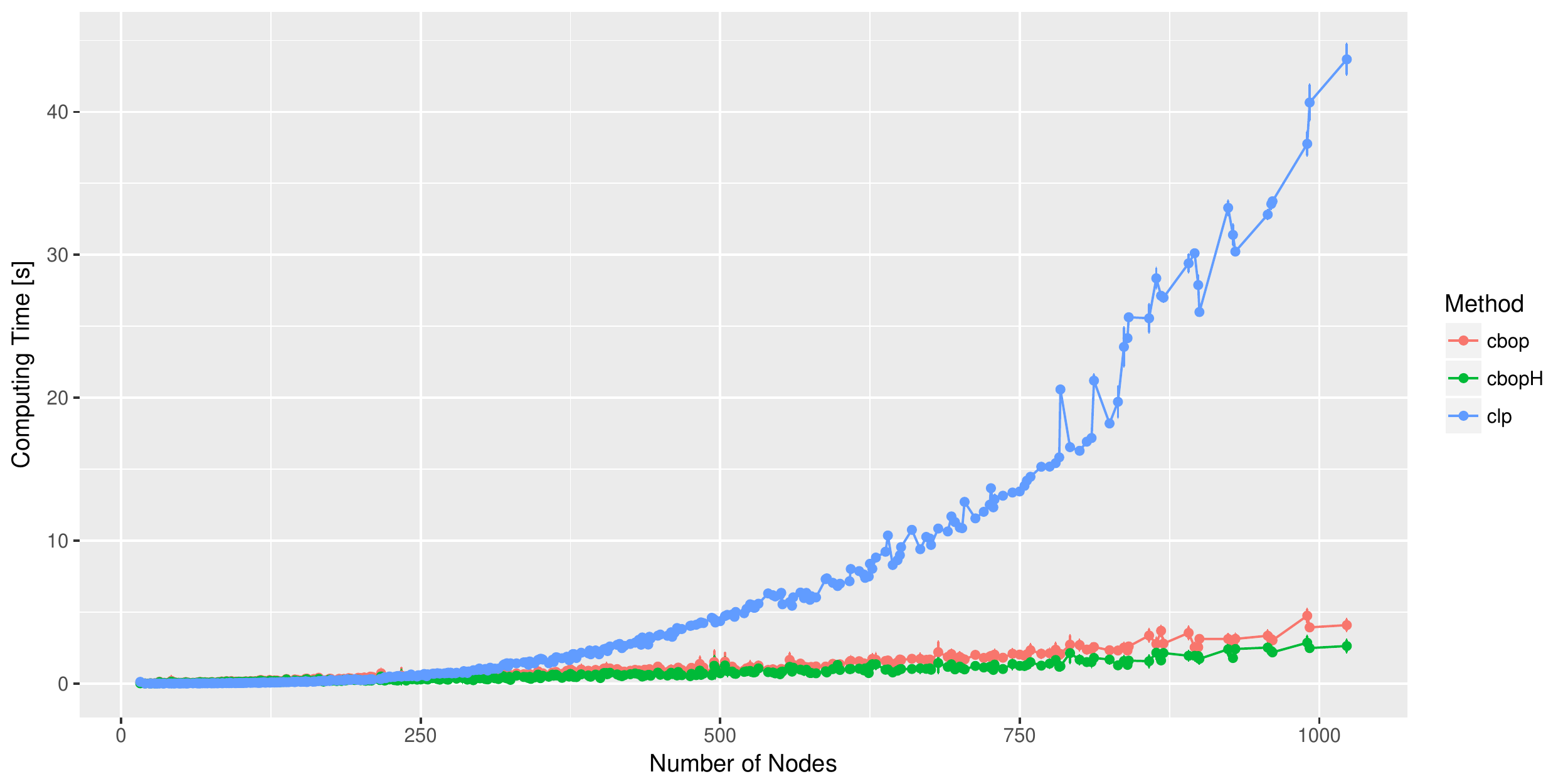}
\end{center}
\caption{Computation times for, respectively, the open-source ``Coin-or" linear programming solver (\texttt{clp}, blue), the margin-constrained bag-of-paths algorithm (\texttt{cbop}, red), and the margin-constrained bag-of-hitting-paths algorithm (\texttt{cbopH}, green) on lattices containing $n$ nodes (x-axis), $\lfloor n/3 \rfloor$ sources and $\lfloor n/3 \rfloor$ targets with uniform weights.}
\label{fig:speed}
\end{figure} 

The results are presented in Figure \ref{fig:speed}. We can observe that on large graphs, both hitting and non-hitting bag-of-paths algorithms perform much faster than the linear programming baseline, with a slight advantage for the hitting algorithm. This was already observed in \cite{Cuturi2013} for entropy regularized optimal transport problems. All results were obtained with Julia (version 0.5.0) running on an Intel Xeon with $2 \times 8$ 3.6GHz processors and 128 GB of RAM.


\section{Conclusion}
\label{Sec_conclusion01}

This work extends the bag-of-paths framework introduced in \cite{mantrach2009sum,franccoisse2016bag} by allowing the user to set constraints on the starting and ending nodes of paths, $\bm{\sigma}_\mathrm{in}$ and $\bm{\sigma}_\mathrm{out}$. Like its predecessor, this formalism is derived for two types of paths, non-hitting paths and hitting paths. It also depends on a user-defined parameter, the temperature $T > 0$, according to which the model interpolates between a deterministic optimal policy and a completely random behavior. Both the non-hitting and hitting paths formalisms allow the computation of various quantities: the coupling, $\gamma_{ij} = \mathrm{P}(S = i, E = j)$; the expected number of times a node appears on the paths (a betweenness value), $\bar{n}_{i}$; and the optimal policy defining a biased random walk with transition probabilities $p_{ij}$. All these quantities are expressed in terms of three computational elements: a fundamental matrix $\mathbf{Z}$, also found in \cite{franccoisse2016bag}, and Lagrange multipliers $\bm{\lambda}_\mathrm{in}$ and $\bm{\lambda}_\mathrm{out}$.

The addition of the set of constraints over starting and ending node distributions adds flexibility to its unconstrained predecessor, and yields interesting connections with other models. When $T \to 0^{+}$, both the non-hitting and hitting formalisms are shown to be similar, and converge to a solution of the optimal transport on a graph problem. Unlike most algorithms solving the transportation problem, both bag-of-paths formalisms not only give sources-to-targets attributions, i.e.\ the coupling $\gamma_{ij}$, but also corresponding embarkment and disembarkment prices (with  $\bm{\lambda}_\mathrm{in}$ and $\bm{\lambda}_\mathrm{out}$) and the flow on edges ($\bar{n}_{ij}$), while running with a competitive computation time compared to a baseline linear solver. In contrast, when $T \to 0$, each formalism behaves differently, each having its own merits. The non-hitting formalism converges to the electrical solution, with starting and ending node distributions corresponding to different potentials defined on nodes, and the hitting formalism, which is faster to compute, converges to the trivial, independent coupling.

These constraints also enlarge the range of applications of the bag-of-paths formalism, and it was shown here how to derive two families of dissimilarities from it. The first family of dissimilarities is defined as the surprisal distance between nodes, and constraints on starting and ending nodes provide a way to associate weights on nodes. The second family of dissimilarities is the free energy dissimilarity between groups of nodes. For the moment, these dissimilarities are quite theoretical and their applications are not explored in this paper. However, future research will investigate the use of these new dissimilarities in semi-supervised classification, hierarchical clustering, as well as other applications.

Generally speaking, the flexibility and the richness of this model could lead to different use cases, and future investigations will aim at finding various applications of the different introduced quantities. An on-going study will also investigate the introduction of flow constraints in the bag-of-paths framework.


\section*{Acknowledgements}

This work was partially supported by the Immediate and the Brufence projects funded by InnovIris (Brussels Region), as well as former projects funded by the Walloon region, Belgium. Ilkka Kivim\"aki was partially funded by Emil Aaltonen Foundation, Finland. We thank these institutions for giving us the opportunity to conduct both fundamental and applied research.

We also thank the anonymous reviewers and the editor whose remarks allowed to improve significantly the manuscript.


\bibliographystyle{apacite}
\bibliography{Randomized_OT}

\end{document}